\documentclass[conference]{IEEEtran}
\usepackage{cite}
\usepackage{amsmath,amssymb,amsfonts, amsthm}
\usepackage{graphicx}
\usepackage{caption}
\usepackage{subcaption}
\usepackage{etoolbox,xspace}
\usepackage{textcomp}
\usepackage{xcolor}
\usepackage{algorithmicx}
\usepackage{algorithm}
\usepackage{algpseudocode}
\algtext*{EndWhile}
\algtext*{EndIf}
\algtext*{EndFunction}
\algtext*{EndFor}
\algrenewcommand\algorithmicrequire{\textbf{Input:}}
\algrenewcommand\algorithmicensure{\textbf{Output:}}

\usepackage{booktabs}

\definecolor{americanrose}{rgb}{1.0, 0.01, 0.24}
\definecolor{airforceblue}{rgb}{0.36, 0.54, 0.66}

\newcommand{\fn}[1]{\textcolor{black}{#1}}

\newcommand{\pranay}[1]{\textcolor{black}{#1}}
\newcommand{\fnrev}[1]{\textcolor{black}{#1}}
\newcommand{\nikolausrev}[1]{\textcolor{black}{#1}}
\newcommand{\pranayrev}[1]{\textcolor{black}{#1}}
\newcommand{\camera}[1]{\textcolor{black}{#1}}
\newcommand{\todo}[1]{{\color{blue} #1}}
\newcommand{\tsim}{\ensuremath{\mathrm{sim}}}
\newcommand{\topklb}{\mathcal{L}_{lb}}
\newcommand{\topkub}{\mathcal{L}_{ub}}
\newcommand{\pqub}{\mathcal{Q}_{ub}}
\newcommand{\sorefine}{{\sc  Refinement}\xspace}
\newcommand{\algo}{{\sc  Koios}\xspace}
\newcommand{\algoplus}{{\sc  Koios+}\xspace}
\newcommand{\sopost}{{\sc  PostProcessing}\xspace}

\DeclareMathOperator*{\argmax}{arg\,max}

\newtheorem{definition}{{\bf{\em Definition}}}

\newtheorem{lemma}{Lemma}

\newcommand{\eop}{\hspace*{\fill}\mbox{$\Box$}}
\newcounter{example}
\renewcommand{\theexample}{\arabic{example}}
\newenvironment{example}{
        \vspace{1ex}
        \refstepcounter{example}
        {\noindent\bf Example \theexample:}}
	{\eop\vspace{1ex}}

\newcommand{\eat}[1]{}

\newcommand{\squishlist}{
 \begin{list}{$\bullet$}
  { \setlength{\itemsep}{0pt}
     \setlength{\parsep}{1pt}
     \setlength{\topsep}{1pt}
     \setlength{\partopsep}{0pt}
     \setlength{\leftmargin}{1em}
     \setlength{\labelwidth}{1em}
     \setlength{\labelsep}{0.5em} } }

\newcommand{\squishend}{
  \end{list}
}

\def\BibTeX{{\rm B\kern-.05em{\sc i\kern-.025em b}\kern-.08em
    T\kern-.1667em\lower.7ex\hbox{E}\kern-.125emX}}

\begin{document}

\title{{\sc Koios}: Top-$k$ Semantic Overlap Set Search}

\author{
\IEEEauthorblockN{Pranay Mundra}
\IEEEauthorblockA{
\textit{University of Rochester}\\
pmundra@ur.rochester.edu}
\and
\IEEEauthorblockN{Jianhao Zhang\textsuperscript{\textsection}}
\IEEEauthorblockA{\textit{Acho Software Inc.} \\
jianhao@acho.io}
\and
\IEEEauthorblockN{Fatemeh Nargesian}
\IEEEauthorblockA{
\textit{University of Rochester}\\
fnargesian@rochester.edu}
\and
\IEEEauthorblockN{Nikolaus Augsten}
\IEEEauthorblockA{
\textit{University of Salzburg}\\
nikolaus.augsten@plus.ac.at}
}



\maketitle
\begingroup\renewcommand\thefootnote{\textsection}
\footnotetext{Work done while at the University of Rochester.}
\endgroup

\begin{abstract}
\pranay{We study the top-$k$ set similarity search problem using semantic overlap. While \fnrev{vanilla} overlap requires exact matches between set elements, {\em semantic overlap} allows \nikolausrev{elements} that are syntactically different but semantically \nikolausrev{related}\eat{similar} to increase the overlap. The semantic overlap is the maximum matching score of a bipartite graph, where an edge weight between two set elements is defined by a user-defined similarity function, e.g., cosine similarity between embeddings.  
\nikolausrev{Common techniques like token indexes fail for semantic search since similar elements may be unrelated at the character level. Further, verifying candidates is expensive (cubic versus linear for syntactic overlap), calling for highly selective filters.}
\eat{For the semantic search, common techniques like token indexes fail as verifying candidates is expensive (cubic vs.\ linear for syntactic overlap).} We propose \algo, the first exact and efficient algorithm for semantic overlap search. \algo leverages sophisticated filters to minimize the number of required graph-matching calculations. Our experiments show that for medium to large sets less than 5\% of the candidate sets need \nikolausrev{verification}, and more than half of those sets are further pruned without requiring the expensive graph matching. 
We show the efficiency of our algorithm on four real datasets and demonstrate the improved result quality of semantic over \fnrev{vanilla set similarity}  search.} 
\end{abstract}


\section{Introduction}
\label{sec:intro}

Set similarity search is a central task  in a variety of applications, such
as data cleaning~\cite{HadjieleftheriouYKS08,WangH19,Yu:2016}, data integration~\cite{0002LECK19,Papadakis2023},  document search~\cite{chen2006similarity},  and dataset discovery~\cite{ZhuNPM16,ZhuDNM19}. 
\fnrev{
The similarity of two sets is typically assessed using vanilla overlap (the number of identical elements of two sets)~\cite{MannAB16,ZhuDNM19} or some normalization of the overlap~\cite{ZhuNPM16,FernandezMNM19}. 
\camera{In the presence of open-world vocabulary and transient quality of data, the vanilla overlap turns out to be ineffective  for sets of strings since it only considers  exact matches between set elements.} 
To address this problem,  fuzzy set similarity search techniques \nikolausrev{like Fast-Join~\cite{5767865,WangLF14} and SilkMoth~\cite{DengKMS17}}
combine 
set similarity and character-based similarity functions on the string set elements, e.g., edit-distance or Jaccard on element tokens. The fuzzy overlap of two sets is  the maximum matching score of a bipartite graph with set elements as nodes and their pairwise character similarity as edge weights.  Unfortunately, fuzzy search can only handle typos and small dissimilarities in set elements and fails for elements that are semantically equivalent or similar but are unrelated at the character level. Since fuzzy set search techniques heavily rely on exact matches between tokens of elements, they cannot be extended to  semantic similarity measures.}

\begin{example} \label{ex:fuzzy}\nikolausrev{Consider query set $Q$ and the collection $\mathcal{L}=\{C_1,C_2\}$ of candidate sets in Fig.~\ref{fig:sym-vs-sem}. The goal is to find the \camera{\mbox{top-$1$}} similar set to $Q$ in $\mathcal{L}$. (1) \emph{Vanilla overlap} considers only the exact match on the set element {\tt\small LA} to assess pairwise set similarities. Typos ({\tt\small Blaine} vs.\ {\tt\small Blain}), synonyms ({\tt\small BigApple} and {\tt\small NewYorkCity}), or other relations between \ elements (e.g., the fact that {\tt \small Charleston} and {\tt \small Columbia} are two cities in South Carolina, {\tt \small SC}) are ignored. (2) \emph{Fuzzy similarity search} allows for matches between syntactically similar set elements. With Jaccard similarity on 3-grams, the relationship between {\tt\small Blaine} to {\tt\small Blain} is detected (3-grams and similarities shown in the figure). However, {\tt\small BigApple} and {\tt\small NewYorkCity} do not contribute to the set similarity; instead, {\tt\small BigApple} is matched to {\tt\small Appleton}, a city in Wisconsin, due to the resemblance of these terms at the character level. Other relationships between set elements are not detected. Therefore, $C_1$ is ranked \camera{\mbox{top-$1$}}, although $C_2$ is more similar to the query: $C_2$ matches on {\tt\small LA} and {\tt\small Blaine} like $C_1$, but in addition has synonyms and semantically related elements. } 
\eat{The vanilla overlap performs poorly: only the number of identical elements) of $C_1$ and $C_2$ is $1$ for both sets. To compute the similarity of $C_1$ and $C_2$ to $Q$, fuzzy set similarity techniques such as Fast-join~\cite{5767865}  and SILKMOTH~\cite{DengKMS17} build a bipartite graph and assign the edit distance or Jaccard similarity of the set of n-grams of elements as  edge weights.  The first column of Fig.~\ref{fig:sym-vs-sem} shows the Jaccard similarity of sets of $3$-grams of elements. Elements with zero Jaccard similarity are not shown. For example, SILKMOTH~\cite{DengKMS17} assigns the  similarity of $3/4$ between the ``Blaine'' in $Q$ and the  typo ``Blain'' in $C_1$ and $C_2$. Interestingly, this technique assigns a similarity score of $1/3$ to {\tt \small BigApple} and {\tt \small Appleton} and a similarity score of $0$ to {\tt \small BigApple} and {\tt \small NewYorkCity}, despite their semantic equivalence. As a result, after computing maximum graph matching of sets, as shown in the middle column of Fig.~\ref{fig:sym-vs-sem}, set search returns $C_1$ as top-$1$ similar set to $Q$. While, in fact in applications such as join search, where $Q$, $C_1$, and $C_2$ are column values in different tables, a data scientist may find $C_2$ more desirable than $C_1$, since it allows joining the record of {\tt \small BigApple} with the record of {\tt \small NewYorkCity} or even joining {\tt \small Charleston} or {\tt \small Columbia} to {\tt \small SC} (two cities in South Carolina state). A search algorithm based on a set similarity measure that takes  into account the semantic similarity of  pairs of elements, even with no character-based similarity, creates the opportunity to boost the similarity of $C_2$ to $Q$. }
\end{example}

\nikolausrev{In this paper, we present {\em semantic overlap} and the \algo algorithm that solves the top-$k$ set similarity search problem using this novel measure. 
Semantic overlap generalizes vanilla and fuzzy overlap and allows semantically related set elements that are unrelated at the character level to contribute to the overall set similarity.}
\fnrev{
Given a query set $Q$  and a candidate set $C$, 
we construct a weighted bipartite graph, where a node is an element in 
$Q$ or $C$, and  an edge between an element of $Q$ and an element of $C$ is weighted by their semantic similarity. The maximum bipartite matching selects a subgraph, such that no two edges share a node and the sum of the subgraph edge similarities is maximized. Following fuzzy set search literature~\cite{DengKMS17,5767865},  
the mapping between set elements is an optional one-to-one mapping. 
The semantic similarity is quantified by a user-specified function, e.g.\camera{,} cosine similarity of embeddings of elements.} 
\begin{figure*}[tbp]
\centering
    \includegraphics[width=\linewidth]{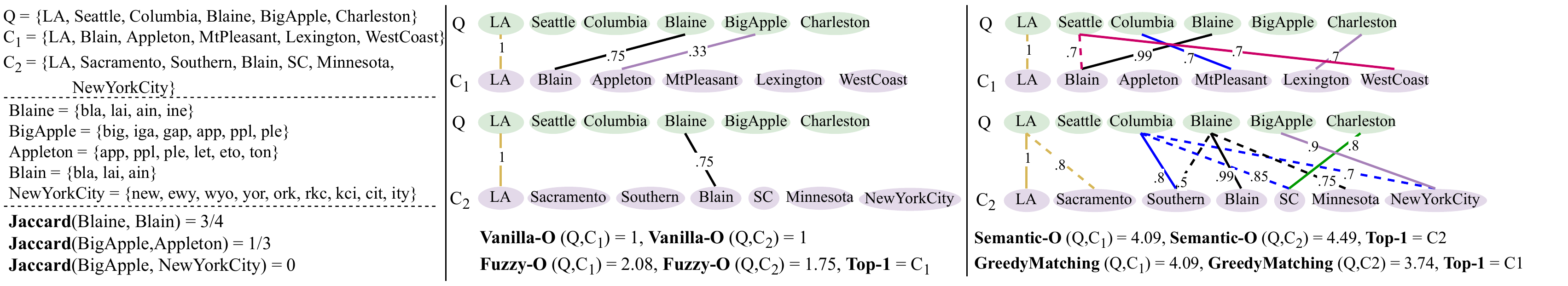}
    \caption{\fnrev{Top-$1$ search using vanilla, fuzzy, and semantic overlap. Element similarity in  fuzzy overlap is the Jaccard of $3$-grams.}}
\label{fig:sym-vs-sem}
\end{figure*}

\begin{example} \label{ex:semvsfuzzy}
\nikolausrev{Continuing Ex.\,\ref{ex:fuzzy}, we perform a top-$1$ search for $Q$ in $\mathcal{L}=\{C_1,C_2\}$ using semantic overlap. Fig.\,1 (on the right) shows the semantic similarities of set elements with a minimum of 0.7 (dashed lines) and the subgraph that maximizes the one-to-one matching (solid lines) and defines the semantic overlap; the matching is optional since not all elements are matched. In addition to elements that are identical or similar at the character level ({\tt\small \camera{Blain}}, {\tt\small Blaine}),  semantically related elements  (e.g.,  {\tt\small Charleston} and {\tt\small SC}) contribute to the semantic overlap; {\tt\small Appleton}, despite its character-level resemblance to {\tt\small BigApple}, does not contribute since it is semantically unrelated. $C_2$ is ranked top-$1$ as expected. Note that a greedy matching approach that matches edges in descending weight order is not optimal and will fail to rank $C_2$ above $C_1$.  }   
\eat{
\fnrev{Continuing with Example~\ref{ex:fuzzy}, suppose a semantic similarity function is used to assign  scores to pairs of elements, as indicated on the edges of graphs in the third column of Fig.~\ref{fig:sym-vs-sem}. Edges with small similarity values are eliminated from the graphs. For example, {\tt \small Blaine} and {\tt \small Blain} are assigned the similarity $0.99$ in both sets and {\tt \small BigApple} and {\tt \small NewYorkCity} are assigned the similarity $0.9$ in $C_2$, while there is no edge between {\tt \small BigApple} and {\tt \small Appleton} due to their lack of semantic similarity. A search based on semantic overlap returns $C_2$ as the top-$1$ result. Edges participating in maximum matching are indicated by solid lines. Note that a greedy bipartite matching algorithm, that approximates semantic overlap by  matching each element to its corresponding non-matched element or none, would estimate the semantic overlap of $4.09$ for $Q$ and $C_1$  and $3.74$ for $Q$ and $C_2$, resulting in top-$1$ becoming $C_1$.}    
}
\end{example}

\fnrev{Semantic overlap lends itself to a wide variety of tasks. For example, in the presence of dirty data and data generated by different standards, \camera{formattings}, and organizations, semantic overlap search can assist with joinable dataset search. Vanilla overlap  has  been extensively studied for table join  search~\cite{ZhuDNM19,FernandezMNM19,ZhuNPM16}.  The notion of semantic join has been explored in SEMA-JOIN~\cite{HeGC15}, were given two joinable columns the goal is to find an optimal way of mapping  values 
in two columns by leveraging the statistical correlation, obtained from a large table corpus, between cell values at both row-level and column-level.   
In addition to discovering joinable columns\camera{,} which is not the focus of SEMA-JOIN, \algo enables finding an optimal way of mapping cell values based on their semantic similarity, when a large corpus of data on cell value mapping  is not available.}

\eat{While vanilla overlap search  has been extensively used for table join  search~\cite{ZhuDNM19,FernandezMNM19,ZhuNPM16}, given the prevalence of web data and open data, generated by different standards, \camera{formattings}, and organizations, the semantic set similarity can be helpful in finding  datasets that are joinable yet not discoverable using traditional (equi-join) ways of joining column values. With respect to applying semantics, existing dataset discovery techniques have only leveraged semantic notions of   similarity for  column names (if existent)~\cite{FernandezAKYMS18} or the aggregated  embedding representation of column values~\cite{NargesianZPM18,FernandezMQEIMO18}. These techniques neither guarantee joinability nor provide the correct matching between column values for joining.} 

\eat{
\fnrev{In addition, this generic definition of overlap lends itself to a wide variety of real-world problems that would  be infeasible with vanilla set overlap.} \pranayrev{In the case of document search, semantic search has been used by augmenting the bag of words of a document with the synonym of document tokens~\cite{?}. Semantic overlap search allows us to efficiently find scientific documents  and research articles by finding more casual matches of technical words, which aid in the discovery of semantically comparable texts. For example, it is usual to use {\tt \small neural networks} and {\tt \small models} interchangeably, but a pure syntactic set search would not regard this pair as a match~\cite{chen2006similarity}. In a supply and demand matching scenario, semantic overlap search would allow for the discovery of suppliers who may be fulfilling demands at the time of emergency, resulting in improved resource allocation~\cite{boysen2019matching}. While using fuzzy techniques can match  product names with typos and small dissimilarities, the semantic overlap search allows matching  names more flexibly, and the one-to-one matching would ensure that there is not a case of short supply due to matching one item  to  multiple items.} }

\eat{\fnrev{There are several challenges to address the semantic set overlap  search problem, which we solve with a novel  filter-verification framework. First, computing semantic overlap is tied to computing the bipartite graph matching of two sets which has the  complexity $\mathcal{O}(n^3)$, where $n$ is the cardinality of sets, and quickly becomes expensive for large sets~\cite{EdmondsK72}. Note that greedy matching which has lower complexity does not consistently achieve optimal top-$k$ results. Second, the sheer number of sets in repositories calls for aggressive  filters to eliminate sets with no potential. The  refinement  phase of \algo avoids  the expensive graph matching of sets at all costs and postpones exact matching to the post-processing phase. Instead, \algo defines bounds for the semantic overlap of sets that are used for filtering and  incrementally and iteratively refines these bounds  during this phase. Third, due to the sheer number of sets and large vocabulary in real repositories, the filters are frequently updated. To alleviate this overhead, \algo supports  efficient-to-update  filters that operate based on the dynamic partitioning of candidate sets. Finally, depending on the distribution of data, many sets may require post-processing.  \algo considers post-processing sets ordered by their potential of being in the top-$k$ results. Moreover, it applies  a specialized filter for the early termination of the graph-matching algorithm, based on the history of  post-processed sets, which  further improves  pruning power.}}
\camera{Several challenges must be addressed to solve the semantic set overlap  search problem. First, computing the semantic overlap requires a bipartite graph matching between two sets, which is expensive and runs in $\mathcal{O}(n^3)$ time for sets with cardinality $n$~\cite{EdmondsK72}. Note that greedy matching, which has lower complexity, does not consistently achieve optimal top-$k$ results. Second, the sheer number of sets in repositories calls for aggressive  filters to eliminate sets with no potential. \algo addresses this issue with a novel filter-verification framework. The  refinement  phase of \algo avoids  the expensive graph matching of sets whenever possible and postpones exact matching to the post-processing phase. \algo defines bounds for the semantic overlap of sets that are used for filtering, and  these bounds are incrementally and iteratively refined. Third, due to the sheer number of sets and large vocabulary in real repositories, the filters are frequently updated. To alleviate this overhead, \algo supports  efficient-to-update  filters that operate based on a dynamic partitioning of candidate sets. Finally, depending on the distribution of data, many sets may require post-processing.  \algo considers post-processing sets ordered by their potential of being in the top-$k$ results. Moreover, it applies  a specialized filter for the early termination of the graph-matching algorithm, based on the history of  post-processed sets, which  further improves the pruning power.}
\pranay{To summarize}, we make the following contributions.  
\squishlist
    \item We propose a new set similarity measure called {\em semantic overlap} that generalizes the vanilla  overlap measure by considering the  semantic similarity of set elements quantified by \eat{arbitrary similarity functions}\camera{a user-defined similarity function}. 
    \item We formulate the top-$k$ set similarity search problem  using semantic overlap and propose a novel  filter-verification framework, called \algo, to address this problem. 
    \item We present powerful and cheap-to-update filters  that aggressively prune sets during both the refinement and post-processing phases. 
    \eat{\item We perform an extensive analysis of the pruning power of filters, response time, and memory footprint on real datasets.  Our experiments show that  \algo is \fn{at least $5.5$x and up to $740$x faster than} baseline that does not use the proposed filters, with a small memory footprint.  \algo performs particularly well for medium to large query sets by pruning more than 95\% of candidate sets.}
    \item \camera{We perform an extensive analysis of the pruning power of filters, response time, and memory footprint on real datasets.  Our experiments show that  \algo has a small memory footprint and is \fn{at least $5.5$x and up to $740$x faster than} a baseline that does not use the proposed filters. \algo performs particularly well for medium to large query sets by pruning more than 95\% of candidate sets.}
\squishend

\section{Problem Definition}
\label{sec:problemdef}

\nikolausrev{We assume sets with pairwise comparable elements,  i.e., the user can define a similarity function for comparing  elements.}  

\begin{definition}[{\bf{\em   Semantic Overlap}}] Given two sets of 
elements $Q$ and $C$ and a similarity threshold \camera{$\alpha > 0$}, 
suppose \fnrev{$M: Q\rightarrow C$ 
is an optional one-to-one matching that determines for each $q_i\in Q$ to be matched to $M(q_i)\in C$ or none. Let $\tsim(.,.)$ be a \nikolausrev{symmetric} similarity function that returns $1$ for identical elements 
and a value in $[0,1]$ for non-identical elements. 
We define 
$\tsim_{\alpha}(x,y) = \tsim(x,y)$, if $\tsim(x,y)\geq\alpha$, otherwise, $0$.} The \emph{semantic set overlap} of $Q$ and $C$ is:  	
\[\mathcal{SO}(Q,C) = \max_M \sum_{q_i\in Q}\tsim_{\alpha}(q_i,M(q_i))\]
\label{def:semoverlap}
\end{definition} 

When set elements are strings, \camera{we refer to set elements as \emph{tokens}}. 
\camera{When clear from the context, we use $\tsim$ for $\tsim_{\alpha}$ and $\mathcal{SO}(C)$ for $\mathcal{SO}(Q,C)$.} 
The semantic overlap \camera{is defined by} the  matching \fnrev{$M : Q \rightarrow C$}
that  maximizes the aggregate pairwise semantic similarity of 
pairs in $M$. 
\pranayrev{The semantic overlap is a symmetric measure.} 
\fnrev{The choice of $\tsim$ 
depends on the context}\nikolausrev{, e.g., if the set tokens are strings,  
the {\em cosine similarity} of their embedding vectors is a common way of comparing elements.} \camera{Other  purely character-based functions include  
the Jaccard similarity of words in  tokens~\cite{DengKMS17} and  the edit distance of tokens~\cite{navarro2001guided,DengKMS17}.}   
\fnrev{
Vanilla overlap is a special case of semantic overlap\eat{, if} \camera{with} $\tsim$ \eat{evaluates}\camera{evaluating} the equality of elements\eat{, i.e.,}\camera{:} $\tsim(q_i,c_j)=1$ if $q_i=c_j$ and $\tsim(q_i,c_j)=0$, otherwise.
\begin{lemma}\label{lemma:vanilla} The vanilla overlap is a lower bound for the semantic overlap\nikolausrev{: $|Q\cap C|\leq\mathcal{SO}(Q,C)$}.
\end{lemma}
\begin{proof}
\eat{\nikolausrev{
Assume $|Q\cap C|>\mathcal{SO}(Q,C)$: We can construct a mapping $M:Q\rightarrow C$ that maps all elements in $|Q \cap C|$ to their identical counterparts. The  similarity sum of this mapping is $|Q \cap C|>\mathcal{SO}(Q,C)$, which contradicts Definition\,\ref{def:semoverlap}.}}
\camera{
$\lvert Q\cap C\rvert>\mathcal{SO}(Q,C)$ is not possible since we can always construct a matching $M:Q\rightarrow C$ that maps all elements in $Q \cap C$ to their identical counterparts. The  similarity sum of this mapping is $\lvert Q \cap C\rvert$.}
\end{proof}}

The semantic overlap of two sets can be formulated as a maximum bipartite graph matching problem. 
For sets $C$ and $Q$, we define a weighted bipartite graph \nikolausrev{$G=(V,E)$, \eat{such that} $V=\{Q, C\}$, where  
 $C$ and $Q$ form disjoint and independent sets of nodes in $G$, and an edge $(q_i,c_j)\in E\subseteq Q\times C$} indicates that $\tsim(q_i,c_j)\neq 0$. The sum of edge weights of a maximum bipartite graph matching is called \camera{the score of the matching.} 
Finding the maximum weighted bipartite graph  matching is known as the  assignment problem and has  time complexity $\mathcal{O}(n^3)$, where $n$ is the cardinality of  sets~\cite{EdmondsK72}. 

\begin{definition}[{\bf{\em Top-k Semantic Overlap  Search}}]
    \label{def:top-k}
     Given a query set $Q$ 
    and a collection of
    sets $\mathcal{L}$,  \nikolausrev{$k\leq|\{C \in \mathcal{L}\mid \mathcal{SO}(Q,C) > 0\}|$}, find a sub-collection $\omega\subseteq \mathcal{L}$ of $k$ distinct sets such that: 
    \begin{enumerate}
	    \item \camera{$\mathcal{SO}(Q,C) > 0 ~\forall C \in \omega$}, and 
    \item \pranay{\camera{$\min\{\mathcal{SO}(Q,C), ~C \in \omega \} \nikolausrev{\geq} \mathcal{SO}(Q,X) ~\forall X \in \mathcal{L}\setminus \omega$}} 
    \end{enumerate}
\end{definition}

\nikolausrev{
Ties are broken arbitrarily so that the result is of size $k$~\cite{IlyasBS08}.} 

\smallskip

\fnrev{Our solution is a filter-verification  algorithm, called \algo.  
In the {\em refinement} phase, \algo identifies candidate sets and  aggressively prunes sets with no potential using cheap-to-apply  filters. During {\em post-processing}, \algo takes pruning to the next level by \nikolausrev{only} partially computing exact matches.} \pranayrev{Fig.~\ref{fig:framework} shows the basic framework of \algo.}

\begin{figure}[tbp]
\centering
    \includegraphics[width=\linewidth]{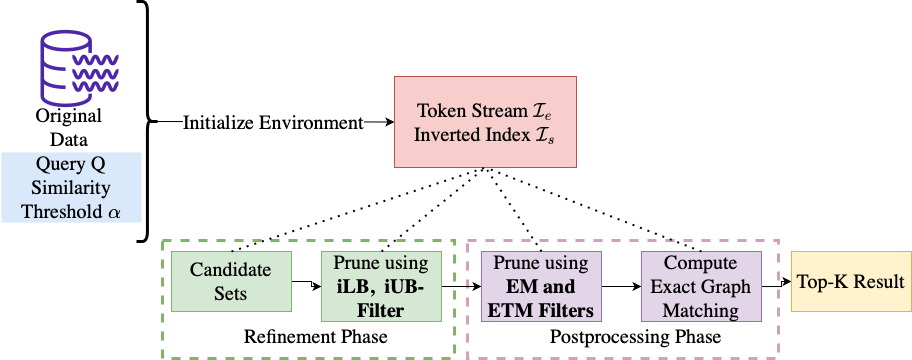}
    \caption{\algo framework.}
\label{fig:framework}
\vspace{-1.5em}
\end{figure}

\eat{\section{\algo : A Two-phase Algorithm}
\label{sec:algorithmv1}

Our solution is  a filter-verification  algorithm.  
The {\em refinement} phase aggressively prunes sets with no potential using cheap-to-apply  filters and the {\em post-processing} takes pruning to the next level by partially computing exact matches. \pranayrev{Figure~\ref{fig:framework} shows the basic framework of \algo.}
}

\eat{
\begin{table}
\tiny
\caption{Table of Notations}
\resizebox{\linewidth}{!}{
\begin{tabular}{ l  l}
\hline
Symbol & Description \\
\hline
$Q$, $C$ & Query set, candidate sets \\
$UB(C)$ & Upper-bound of $\mathcal{SO}$ of candidate set $C$\\
$LB(C)$ & Lower-bound of $\mathcal{SO}$ of candidate set $C$\\
$\mathcal{L}$ & \nikolausrev{Collection of sets}\\
$\mathbb{D}$ & Dictionary of $\mathcal{L}$ \\
$\mathcal{I}_e$ & Token stream \\
$\topklb$ & Top-$k$ LB list\\
$\topkub$ & Top-$k$ UB list\\
$\pqub$ & Priority queue on UBs \\
 \hline
\end{tabular}
}
\label{tbl:notations}
\end{table}
}

\section{Basic Filtering}
\label{sec:basicfilter}

We refer to the list of result sets partially ordered by  semantic overlap  descendingly as {\em top-$k$ result} and denote $\theta^*_k$ as the semantic overlap of the set  with the $k$-th smallest semantic overlap in the top-$k$ result. 
Of course, the value of $\theta^*_k$ is not known until the top-$k$ result \nikolausrev{is computed.}
\eat{Assume we iterate over sets in $\mathcal{L}$ (in an effective way that allows skipping over sets) and maintain a {\em running top-$k$ list} containing the best sets the algorithm has found so far. }
\nikolausrev{Assume an algorithm that iterates over sets in  $\mathcal{L}$  and maintains a {\em running top-$k$ list} containing the best sets found so far:}
The running top-$k$ may contain any $k$-subset of $\mathcal{L}$ with non-zero \camera{semantic overlap}. 
\nikolausrev{We denote the smallest semantic overlap of sets in the running top-$k$ list by  $\theta_k$; by \fnrev{definition, \eat{we have} $\theta_k\leq\theta^*_k$.} }

Based on Def.~\ref{def:semoverlap}, any set $C\in \mathcal{L}$ that contains at least one  element with  similarity higher than $\alpha$ to some  element of $Q$ has a non-zero  semantic overlap and is a candidate set. \fnrev{
Clearly, if $\mathcal{SO}(C)<\theta_k$, set $C$ can be pruned, as there exist at least $k$ sets with better $\mathcal{SO}$ scores. However, since computing $\mathcal{SO}(Q,C)$ is expensive (cubic in set cardinality), our goal during the refinement is to \nikolausrev{prune without computing} the exact matching.  \nikolausrev{To this end,} we compute  lower and upper bounds of the semantic overlap for candidate sets. The bounds help prune sets by comparing against the current $\theta_k$, hence, reducing the number of exact graph matching  calculations.}


\noindent{\bf UB-Filter:} \fnrev{We define an upper bound for semantic overlap.
\begin{lemma} \label{lemma:ub} \nikolausrev{Given query set $Q$ and candidate set $C$, then  \camera{$\mathcal{SO}(C)\leq|Q|\cdot \max_{q_i\in Q, c_j\in C}\{\tsim(q_i,c_j)\} = UB(C)$}.} 
\end{lemma}
\begin{proof} The weight of an edge, $(q_i,c_j)\in Q\times C$, in the semantic overlap bipartite graph of sets $Q$ and $C$ is bounded by the maximum similarity between any two elements of the sets. 
\nikolausrev{The size of any matching $M$ between $Q$ and $C$ is bound by $|M|\leq min(|Q|,|C|)\leq |Q|$. \camera{With Def.~\ref{def:semoverlap}, $\mathcal{SO}(C)\leq UB(C)$.}}
\end{proof}
\eat{Let $UB(C) = |Q|\cdot \max_{q_i\in Q, c_i\in C}\{\tsim(q_i,c_j)\}$.} \camera{Since $UB(C)\leq\mathcal{SO}(\nikolausrev{C})$, any set $C$ that satisfies $UB(C)<\theta_k$ can be safely pruned \nikolausrev{with Lemma~\ref{lemma:ub}}.} 
}


\noindent{\bf LB-Filter:} \fnrev{\camera{Given a bipartite graph, \eat{at each iteration,} the greedy matching algorithm  at each iteration includes the edge with the highest weight  between unmatched nodes  until no edge with  unmatched nodes can be found.} The greedy algorithm for maximum matching has complexity \camera{$\mathcal{O}(n^2\cdot \log n)$}, where $n$ is the cardinality of sets. As shown in Ex.~\ref{ex:semvsfuzzy}, the  greedy algorithm does not find the optimal solution. \nikolausrev{However, the score of the greedy matching has been shown  to be at least half of the optimal score~\cite{algobook}.} 
\begin{lemma}\label{lemma:lb} \nikolausrev{Let $Q$ and $C$ be two sets with semantic overlap $\mathcal{SO}(C)$. In the corresponding bipartite graph, let $LB(C)$ be the maximum of  
(a) the  maximum  edge weight, and (b) the greedy matching score. Then, $LB(C)\leq \mathcal{SO}(C)$.}
\end{lemma}
\begin{proof} 
\eat{
The optimal 
semantic overlap bipartite graph matching of two sets contains at least the 
edge with the maximum weight. }
\nikolausrev{
(a) We can always construct a one-to-one matching $M=\{(q_i, c_j)\}$ that consists of the edge $(q_i, c_j)$ with maximum weight. (b) Since the greedy matching of a bipartite graph is a lower bound of the optimal matching, $\mathcal{SO}(C)$ is lower-bounded by the greedy matching score.}   
\end{proof}
}  

If we knew the value of $\theta^*_k$\camera{,} we could do the maximum pruning during the refinement step. 
\eat{Choosing a $\theta_k$ that is close to the absolute value of $\theta_k^*$ helps with pruning more sets using the UB-Filter.}\camera{Initializing $\theta_k$ with a value close to $\theta_k^*$ will improve the prunig power of the UB-Filter.} One way is to initialize the top-$k$ list by computing the $\mathcal{SO}$ \camera{of a sample of sets} and picking the top-$k$ ones. In this case, the gained pruning power of the UB-Filter comes at the cost of graph matching \camera{calculations}. To avoid this cost, we  leverage a lower bound of $\theta_k$ for pruning.  
\fnrev{\begin{lemma}\label{lemma:thetalb} \camera{Let $\mathcal{R}$ be the running top-$k$ list and $\theta_{lb}$ the 
smallest lower bound $LB(C)$ for  $C\in\mathcal{R}$. Then, $\theta_{lb}\leq\theta^*_k$.} 
\end{lemma}
\begin{proof} \eat{\nikolausrev{By definition,} $\theta_{k,lb}$ is the minimum $LB$ of sets in 
$\mathcal{R}$, and  $\theta_{k}$ is the minimum exact $\mathcal{SO}$ of sets in $\mathcal{R}$.  \nikolausrev{Since $LB(C)\leq\mathcal{SO}(C)$ for any set $C\in\mathcal{L}$, $\theta_{k,lb}=min_{C\in\mathcal{R}}\{LB(C)\}\leq min_{C\in\mathcal{R}}\{\mathcal{SO}(C)\}=\theta_k$ such that   $\theta_{k,lb}\leq\theta_k\leq \theta^*_k$.}}
\camera{By definition, $\theta_{lb}$ is the minimum $LB$ of sets in 
$\mathcal{R}$, and  $\theta_{k}$ is the minimum exact $\mathcal{SO}$ of sets in $\mathcal{R}$. Since $LB(C)\leq\mathcal{SO}(C)$ for any set $C\in\mathcal{L}$, $\theta_{lb}=min_{C\in\mathcal{R}}\{LB(C)\}\leq min_{C\in\mathcal{R}}\{\mathcal{SO}(C)\}=\theta_k$ such that $\theta_{lb}\leq\theta_k\leq \theta^*_k$.}
\end{proof}
}
\fnrev{With Lemma~\ref{lemma:thetalb}, we can safely prune \camera{a candidate set} \eat{sets} if $UB(C)<\theta_{lb}$}.   

\section{Refinement: Candidate Selection}
\label{sec:candselection}

\eat{Similarly, $\theta_{k,ub}$ is the $k$-th largest $UB(C)$, $\forall C\in\mathcal{L}$.  
When clear from the context we refer to $\theta_{k,lb}$ and  $\theta_{k,ub}$ with  $\theta_{lb}$ and $\theta_{ub}$, respectively. 
The algorithm applies $\theta_{ub}$ to prune sets during exact match calculations of the post-processing phase. 
Note that, when we calculate the  semantic overlap of a set, we have $LB(C) = UB(C) = \mathcal{SO}(C)$. Therefore, the exact values of $\mathcal{SO}(Q)$ can contribute to both $\theta_{lb}$ and $\theta_{ub}$. During the post-processing when exact semantic overlaps are calculated, $\theta_{ub}$ gets updated until it converges to the true $\theta^*_k$.  
To evaluate $\theta_{lb}$, we need to maintain a running ordered list of top lower-bounds. The lower-bounds gradually increase  during the refinement phase. Whenever the lower-bound of a set becomes higher than the current $\theta_{lb}$, the list is updated and $\theta_{lb}$ is updated. This results in potentially replacing a set in the running top-$k$  list with a new set. 
On the other hand, the upper bounds of sets gradually decrease  during the refinement phase. Therefore, we need to maintain a data structure that keeps track of the upper bounds of all sets in the candidate collection. 
Upon an update to the upper bound of a set $C$, if $UB(C)$ is larger than the current $\theta_{lb}$, we prune  set $C$ from the upper-bound data structure.  
\begin{figure}[tbp]
\centering
    \includegraphics[width=0.7\linewidth]{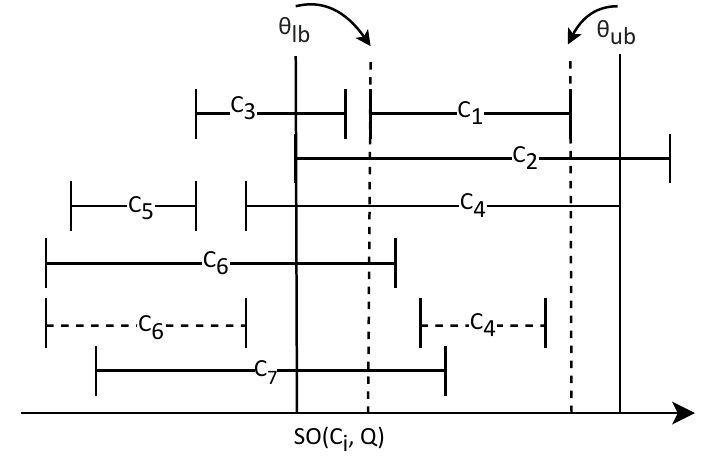}
    \caption{Pruning sets with upper and lower bounds for $k=2$. Each set is represented with  its LB and UB. }
\label{fig:intervals}
\end{figure}

\begin{example}  Suppose sets $C_1, \ldots, C_7$ of  Figure~\ref{fig:intervals} are in the post-processing phase. Each set is represented with an  interval of its lower- and upper-bound. Note that $\theta_{ub}$ is considered for\eat{the purpose of} this example. \algo \eat{\algoplus} works with $\theta_{lb}$ and it is only during the post-processing phase that $\theta_{ub}$ is populated. Suppose we are looking for the top-2 sets with the highest semantic overlap with a query.  Suppose the  value of $\theta_{lb}$ is initially calculated based on the lower-bounds of $C_1$ and $C_2$, since they have the top-2 lower bounds among all sets. As a result, set $C_5$ is pruned because $UB(C_5)<\theta_{lb}$. The remaining sets stay in the candidate collection because they all have lower-bounds smaller than  $\theta_{lb}$ and upper-bounds that are greater than $\theta_{lb}$.  
Suppose, at the next iteration, $LB(C_4)$ and $UB(C_4)$ are  updated. Now, $C_4$ is  the set with the highest lower bounds and as a result, the value of $\theta_{lb}$ is updated to $LB(C_1)$ and the value of $\theta_{ub}$ is updated to $UB(C_1)$. This allows us to safely prune $C_3$ because now  $UB(C_3)<\theta_{lb}$. In the next iteration, if an update to  $UB(C_6)$ happens and $UB(C_6)$ becomes smaller than $\theta_{lb}$, we can prune $C_6$ from the candidate collection. At the end of the refinement phase, the algorithm passes all remaining sets \fnrev{($C_1$, $C_2$, $C_4$, and $C_7$)} to the post-processing phase. \label{ref:intervals1} 
 \end{example}

\begin{figure}[tbp]
\centering
    \includegraphics[width=0.95\linewidth]{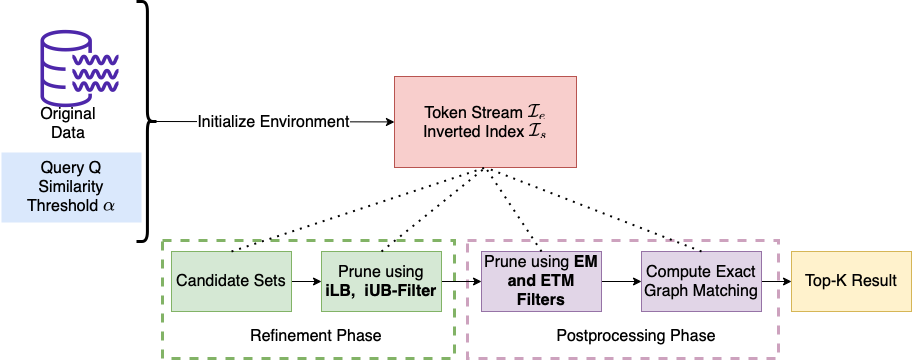}
    \caption{\algo Framework}
\label{fig:framework}
\end{figure}
}


\sloppy{\camera{In the refinement phase, we use two index structures: the token stream $\mathcal{I}_e$ and the inverted index $\mathcal{I}_s$.} Let $\mathbb{D}=\cup_{C_j\in\mathcal{L},c_i\in C_j}~c_i$, be the vocabulary of $\mathcal{L}$. 
\eat{Suppose we have a way of getting a stream of all elements in $\mathbb{D}$ 
ordered in descending 
order of the maximum  similarity to the set of query elements. We call this stream a {\em token stream}, namely $\mathcal{I}_e$.}\camera{The token stream $\mathcal{I}_e$ is a stream of all elements in $\mathbb{D}$ ordered in descending order by the maximum similarity to any query element $q_i \in Q$.} The token stream  \camera{$\mathcal{I}_e$ is a sequence of tuples $(q_i,c_j,\tsim(q_i,c_j))$, where $q_i\in Q$ and $c_j\in\mathbb{D}$.} If $\tsim(q_i,c_j)<\tsim(q_l,c_m)$, 
\eat{element}\camera{tuple} \fnrev{$(q_i,c_j,\tsim(q_i,c_j))$} will \camera{follow}  $(q_l,c_m,\tsim(q_l,c_m))$ \camera{in $\mathcal{I}_e$}. The stream stops when there is no token $c_j\in\mathbb{D}$ \camera{left with $\tsim(q_i,c_j)\geq\alpha, \eat{\forall} q_i\in Q$.}\eat{We will describe how this stream is  created shortly.} 
\eat{An} \camera{The inverted index $\mathcal{I}_s$ \eat{which}maps $c_j\in\mathbb{D}$ to $\{C_1,\ldots,C_m\} \subseteq \mathcal{L}$, such that $\forall~C_i, 1\leq i\leq m : c_j\in C_i$.} 
Upon reading a tuple $(q_i,c_j,\tsim(q_i,c_j))$ from $\mathcal{I}_e$, the index $\mathcal{I}_s$ is probed to obtain all sets containing $c_j$. This creates a stream of sets in $\mathcal{L}$, in descending 
order of the maximum  similarity set elements to \camera{some query element.}\eat{query elements.} Clearly, if a set  appears in this stream it has a non-zero  semantic overlap and is a candidate set. \fnrev{The first time we observe $C$ in this stream, $(q_i,c_j,\tsim(q_i,c_j))$ is in fact $\max\{\tsim(q_i,c_j)\}, c_j\in C, q_i\in Q$. Therefore, based on Lemma~\ref{lemma:ub} and~\ref{lemma:lb}, 
we can initialize  $UB(C) = min(|Q|,|C|)\cdot \tsim(q_i,c_j)$ and $LB(C) = \tsim(q_i,c_j)$. \algo uses the first $k$ sets obtained from  \eat{the}\camera{$\mathcal{I}_e$} to initialize the running top-$k$ list and $\theta_{lb}$. As more sets are obtained, sets with lower $LB(C)$ may be found which results in updating the running top-$k$ list and $\theta_{lb}$.} 
}

Algorithm~\ref{alg:refinement} presents the pseudo-code of the refinement phase of \algo. \eat{In Section~\ref{sec:refinement}, we describe stronger filters based on tighter bounds.} The power of the algorithm is in  postponing exact match calculation to the post-processing phase\camera{,} while pruning sets aggressively in this phase. In each iteration, \algo  reads tuples from  $\mathcal{I}_e$ and updates the bounds of candidate sets. Any set $C$ with \eat{initial}$UB(C)<\theta_{lb}$ is safely pruned. A set $C$ with $LB(C)\leq\theta_{lb}\leq UB(C)$ is in limbo until more elements are read and more evidence about the bounds of the set is collected. The changes could be\camera{:} $UB(C)$ decreases, $LB(C)$ increases, or $\theta_{lb}$ increases. 
\fnrev{\eat{Designing streams}\camera{The indexes} $\mathcal{I}_e$ and $\mathcal{I}_s$ allow us to obtain  candidate sets in the descending order of their initial upper- and lower-bounds. 
\nikolausrev{By processing candidate sets in this order, the search algorithm identifies promising sets early, can update the top-$k$ list, and improve $\theta_{lb}$ to achieve a high pruning ratio.} }

\eat{A candidate set with  $LB(C)\geq\theta_k$ replaces an existing set in the running top-$k$ list. \pranayrev{Note that the greedy matching score used for $LB(C)$ is only an approximation of the true  semantic overlap. To get the exact semantic overlap, we need to evaluate the exact matching of $C$ and $Q$, which is expensive for  large sets.} \eat{A set with $LB(C)\geq\theta_k$ may or may not be in the final top-$k$ list, however, calculating the exact semantic overlap of a set with $LB(C)\geq\theta_k$ has the benefit of increasing the current $\theta_k$, which results in pruning additional sets based on their upper bounds. Because the exact match calculation is expensive, our}\pranayrev{Our algorithm postpones the exact evaluation of semantic overlap until the bounds of all sets are tightened through the token stream, that is no tokens with similarity higher than $\alpha$ to some query token are left in the token stream. Only then, the algorithm starts on the exact calculation of the semantic overlap for sets with $\theta_k\leq UB(Q)$. Algorithm~\ref{alg:refinement} presents the pseudo-code of the refinement phase of the top-$k$ semantic overlap search.}}


\fnrev{\eat{Since based on}\camera{Due to} Lemma~\ref{lemma:ub}, $UB(C)$ is tied to its largest element similarity to \camera{some query element.}\eat{query elements.}} \nikolausrev{To} order sets based on initial $UB$ and $LB$,  we need to retrieve elements in the vocabulary $\mathbb{D}$ in descending order of their similarity to \camera{some query element.}\eat{query elements.} To avoid computing all pairwise similarities between vocabulary and query elements, for a given $\tsim$, any index that enables \camera{efficient threshold-based similarity search}  
is suitable. For example, when $\tsim$ is the cosine similarity  of word embeddings of tokens or the Jaccard of the token set of elements, the Faiss Index~\cite{johnson2019billion} or minhash LSH~\cite{Broder97} can be plugged into the algorithm, respectively. This allows \algo to  perform semantic overlap search independent of the choice of $\tsim$. \nikolausrev{This sets \algo apart from } \fnrev{existing fuzzy set similarity search techniques~\cite{DengKMS17,AgrawalAK10,DengLFL13,WangLDZF15}, which \nikolausrev{rely on} \eat{signatures}\pranayrev{filters} designed for specific  similarity functions.} 

To retrieve   candidate sets in the descending order of their  upper-bounds, a naive solution  is to extend the idea of the inverted index with a map that associates each query element to a list of all elements in the vocabulary. The elements in a list are ordered descendingly according to the similarity of an element to the query element. 
The size of this index grows linearly with the cardinality of $Q$. 
Now, let $r: Q \to \mathbb{A}$ be given by $r_j(q) = t^j$ 
for $j \in {1, \ldots, |\mathbb{D}|}$ that is to say that 
$r_j(q)$ returns the $j$-th most similar element of $\mathbb{D}$ 
to  $q\in Q$. 
Note that the $j$-th most similar element to elements of $Q$ is not necessarily the one with $j$-th highest similarity to any element  in $Q$, 
$\argmax_{t^j} (\tsim(t^j,q_i),\forall i\in [Q])$. For example, $r_{i-1}(q)$ and 
$r_i(q)$ can both have higher similarities than  $r_{1}(q^\prime)$.  
The function $r$ can be realized by one shared index $\mathcal{I}$ over $\mathbb{D}$ and a priority queue $\mathcal{P}$ of size $|Q|$ that keeps track of the most similar elements to elements of $Q$. 
We refer to  $\mathcal{I}$ and $\mathcal{P}$ as  {\em token stream}   $\mathcal{I}_e$. 

Given an element $q$, the index $\mathcal{I}$ returns the next most similar unseen element of $\mathbb{D}$ to $q$.   
In the initial step, we probe $\mathcal{I}$ with all elements in $Q$ and add results to the  priority queue. At this point, $\mathcal{P}$ contains  $|Q|$ elements,  each being the most similar  element to a query element. The queue keeps track of the query element corresponding to each element in $\mathcal{P}$. 
\pranay{The {\em top} of the queue always gives us the mapping of the query element to the unseen element with the highest similarity.}
The second most similar element to $Q$ is already buffered in $\mathcal{P}$.  
To maintain the queue size, when we pop the top element, we only require to probe $\mathcal{I}$ with the query element  corresponding to the popped element because the most similar elements to the rest of the query elements still exist in $\mathcal{P}$.    
Note that when a non-zero $\alpha$ threshold is provided, we stop probing $\mathcal{I}$ when the first  element with a similarity smaller than $\alpha$ is retrieved. 

In addition, we build an  inverted index 
that maps each element in $\mathbb{D}$ to the corresponding sets in $\mathcal{L}$. 
Since $\mathcal{I}_e$ returns elements in the descending order of similarity to $Q$, the new sets that are retrieved from $\mathcal{I}_s$ 
arrive at the descending order of UB-Filter.   
\eat{The space complexity of $\mathcal{I}_e$ is in the size of the vocabulary, $|\mathbb{D}|$,  plus the size of the query,  $|\mathbb{D}|+|Q|$. 
The space complexity of $\mathcal{I}_s$ is in the size of the vocabulary times the number of sets in $\mathcal{L}$, $|\mathbb{D}|\times|\mathcal{L}|$.}  

\eat{
\subsection{{\bf{\em Bounds of $\mathbf{\theta_k}$}}} We know   $\theta^*_k \geq \theta_k$. 
Note that if we knew the true value of $\theta^*_k$ we could do the maximum pruning during the refinement step. 
Because exact matching is expensive, we  apply a lazy approach and postpone calculating $\theta_k$ and refining it to become $\theta^*_k$ to the post-processing phase. 
We would still like to define a lower bound for $\theta_k$ and leverage that for pruning.  
We know that the exact semantic overlap value of a set   is not smaller  
than its $LB(C)$. A set may or may not end up in the final top-$k$ list.  
We define $\theta_{k,lb}$ as the $k$-th largest $LB(C)$, $\forall C\in\mathcal{L}$. 
During the refinement phase, we prune sets based on the $\theta_{k,lb}$, 
i.e. sets with $UB(C)<\theta_{k,lb}$ are safely pruned. 
Similarly, $\theta_{k,ub}$ is the $k$-th largest $UB(C)$, $\forall C\in\mathcal{L}$.  
When clear from the context we refer to $\theta_{k,lb}$ and  $\theta_{k,ub}$ with  $\theta_{lb}$ and $\theta_{ub}$, respectively. 
The algorithm applies $\theta_{ub}$ to prune sets during exact match calculations of the post-processing phase. 
Note that, when we calculate the  semantic overlap of a set, we have $LB(C) = UB(C) = \mathcal{SO}(C)$. Therefore, the exact values of $\mathcal{SO}(Q)$ can contribute to both $\theta_{lb}$ and $\theta_{ub}$. During the post-processing when exact semantic overlaps are calculated, $\theta_{ub}$ gets updated until it converges to the true $\theta^*_k$.  

To evaluate $\theta_{lb}$, we need to maintain a running ordered list of top lower-bounds. The lower-bounds gradually increase  during the refinement phase. Whenever the lower-bound of a set becomes higher than the current $\theta_{lb}$, the list is updated and $\theta_{lb}$ is updated. This results in potentially replacing a set in the running top-$k$  list with a new set. 

On the other hand, the upper bounds of sets gradually decrease  during the refinement phase. Therefore, we need to maintain a data structure that keeps track of the upper bounds of all sets in the candidate collection. 
Upon an update to the upper bound of a set $C$, if $UB(C)$ is larger than the current $\theta_{lb}$, we prune  set $C$ from the upper-bound data structure.  
\begin{figure}[tbp]
\centering
    \includegraphics[width=0.7\linewidth]{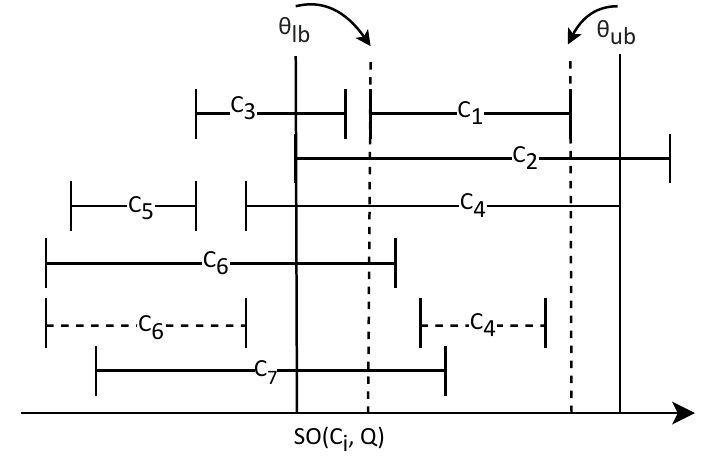}
    \caption{Pruning sets with upper and lower bounds for $k=2$. Each set is represented with  its LB and UB. }
\label{fig:intervals}
\end{figure}

\begin{example}  \fnrev{Suppose sets $C_1, \ldots, C_7$ of  Figure~\ref{fig:intervals} are in the post-processing phase.} Each set is represented with an  interval of its lower- and upper-bound. Note that $\theta_{ub}$ is considered for\eat{the purpose of} this example. \algo \eat{\algoplus} works with $\theta_{lb}$ and it is only during the post-processing phase that $\theta_{ub}$ is populated. Suppose we are looking for the top-2 sets with the highest semantic overlap with a query.  Suppose the  value of $\theta_{lb}$ is initially calculated based on the lower-bounds of $C_1$ and $C_2$, since they have the top-2 lower bounds among all sets. As a result, set $C_5$ is pruned because $UB(C_5)<\theta_{lb}$. The remaining sets stay in the candidate collection because they all have lower-bounds smaller than  $\theta_{lb}$ and upper-bounds that are greater than $\theta_{lb}$.  
Suppose, at the next iteration, $LB(C_4)$ and $UB(C_4)$ are  updated. Now, $C_4$ is  the set with the highest lower bounds and as a result, the value of $\theta_{lb}$ is updated to $LB(C_1)$ and the value of $\theta_{ub}$ is updated to $UB(C_1)$. This allows us to safely prune $C_3$ because now  $UB(C_3)<\theta_{lb}$. In the next iteration, if an update to  $UB(C_6)$ happens and $UB(C_6)$ becomes smaller than $\theta_{lb}$, we can prune $C_6$ from the candidate collection. At the end of the refinement phase, the algorithm passes all remaining sets \fnrev{($C_1$, $C_2$, $C_4$, and $C_7$)} to the post-processing phase. \label{ref:intervals1} 
 \end{example}
 }

\setlength{\textfloatsep}{0.1cm}
\begin{algorithm}
\caption{\algo(\sorefine)}
\begin{footnotesize}
\begin{algorithmic}[1]
    \renewcommand{\algorithmicrequire}{\textbf{Input:}}
    \renewcommand{\algorithmicensure}{\textbf{Output:}}
    \Require $\mathcal{L}$: a repository, $\mathcal{I}_{e}$: \eat{a similarity index on tokens of $\mathcal{L}$}\camera{token stream}, $\mathcal{I}_{s}$: an inverted index on sets in $\mathcal{L}$, $Q$: a query set, $k$: search parameter, $\alpha$: token similarity threshold
    \Ensure $\mathcal{U}$: \eat{unpruned}\camera{candidate} sets
    \State $\mathcal{U}\gets\emptyset$, $\mathcal{P}\gets\emptyset$ // candidate sets, pruned sets
	\State $sim\gets 1$ , $\topklb\gets\emptyset$ // initial similarity, top-$k$ LB list
    \While{$sim \geq \alpha$}
	    \State $t,sim\leftarrow{\tt get\_next\_similar\_token(Q, \mathcal{I}_{e})}$
	    \If{$sim \geq \topklb.{\tt bottom}()$} $Cs\leftarrow{\tt get\_sets(t, \mathcal{I}_{s})}$
            \For{$C\in Cs$ and $C\notin\mathcal{U}$ and $C\notin\mathcal{P}$} $\mathcal{U}.{\tt add}$(C)
        	\EndFor
		\EndIf
		\For{$C\in\mathcal{U}$}
    		\State ${\tt UB}(C).{\tt update}(t, sim)$
		    \If{${\tt UB}(C)<\topklb.{\tt bottom}()$}
    			\State $\mathcal{U}.{\tt remove}(C)$, $\mathcal{P}.{\tt add}(C)$
    	    \EndIf
    		\State	$LB(C).{\tt update}(t,~sim)$
		    \If{$LB(C) > \topklb.{\tt bottom}()$} $\topklb.{\tt update}(C)$
		    \EndIf
		\EndFor
	\EndWhile
    \Return $\mathcal{U}$ 
\end{algorithmic}
\end{footnotesize}
\label{alg:refinement}
\end{algorithm}
\setlength{\floatsep}{0.1cm}

\eat{
\subsection{{\bf{\em Ordering Candidate Sets}}}
\label{sec:index}
\eat{\fn{R3-O7: add "to avoid  computing similarity between the cross-product of the query tokens and all tokens in the database, ..."}\pranay{addressed}}
\eat{We start by presenting  a technique for retrieving tokens in the vocabulary $\mathbb{D}$ in the descending order of their similarity to query tokens. \pranay{To avoid computing all pairwise similarities between the query tokens and all tokens in the database, we use the Faiss Index, which also provides for a threshold-based search~\cite{johnson2019billion}.}}
We start by presenting a technique for retrieving tokens in the vocabulary $\mathbb{D}$ in descending order of their similarity to query tokens. To avoid computing all pairwise similarities between vocabulary and query tokens, for a given $\tsim$, any index that enables similarity threshold-based search in sub-linear time, i.e., get all tokens in $\mathbb{D}$ with similarity to the query tokens greater than the threshold, is suitable. For example, when $\tsim$ is the cosine similarity  of word embeddings of tokens or the Jaccard of token set of  elements, the Faiss Index~\cite{johnson2019billion} or minhash LSH~\cite{Broder97} can be plugged in the algorithm, respectively. This allows \algo to  perform semantic overlap search independent of the choice of $\tsim$ which  is unlike the existing fuzzy set similarity search techniques~\cite{DengKMS17,AgrawalAK10,DengLFL13,WangLDZF15} which are designed for specific token similarity functions. 

Building upon this technique and using an inverted index that maps each token to the sets that contain it, we present a way of retrieving  candidate sets in the descending order of their initial upper bounds. 
Processing candidate sets in this order enable the search to identify sets that potentially make it to the top-$k$ list early during the search, improve $\theta_{lb}$,  and achieve a high pruning ratio.

A naive solution to retrieve sets in the descending order of their initial upper bounds (UB-Filter) is to extend the idea of the inverted index with a map that associates each query token to a list of all tokens in the vocabulary. The tokens in a list are ordered descendingly according to the similarity of a token to the query token. Formally,  a token $q\in Q$ is mapped to set 
$\mathbb{A} = \{t^1, \ldots, t^{|\mathbb{D}|}\}$, i.e. the tokens in $\mathbb{D}$ are labeled such that 
$
\tsim(t^1,q) \geq \tsim(t^2,q) \geq \cdots \geq \tsim(t^{|\mathbb{D}|},q)$.  

The size of this index grows linearly with the cardinality of $Q$. 
Now, let $r: Q \to \mathbb{A}$ be given by $r_j(q) = t^j$ 
for $j \in {1, \ldots, |\mathbb{D}|}$ that is to say that 
$r_j(q)$ returns the $j$-th most similar token of $\mathbb{D}$ 
to token $q\in Q$. 
Note that the $j$-th most similar token to tokens of $Q$ is not necessarily the token with $j$-th highest similarity to any token in $Q$, 
$\argmax_{t^j} (\tsim(t^j,q_i),\forall i\in [Q])$. For example, $r_{i-1}(q)$ and 
$r_i(q)$ can both have higher similarities than  $r_{1}(q^\prime)$.  
The function $r$ can be realized by one shared index $\mathcal{I}$ over $\mathbb{D}$ and a priority queue $\mathcal{P}$ of size $|Q|$ that keeps track of the most similar tokens to tokens of $Q$.  In Algorithm~\ref{alg:refinement}, we refer to  $\mathcal{I}$ and $\mathcal{P}$ with {\em token stream}   $\mathcal{I}_e$. 
Given a token $q$, the index $\mathcal{I}$ returns the next most similar unseen token of $\mathbb{D}$ to $q$.   
In the initial step, we probe $\mathcal{I}$ with all tokens in $Q$ and add results to the  priority queue. At this point, $\mathcal{P}$ contains  $|Q|$ tokens each is the most similar  token to a query token. The queue keeps track of the query token corresponding to each token in $\mathcal{P}$. 
\pranay{The {\em top} of the queue always gives us the mapping of the query token to the unseen token with the highest similarity.}
The second most similar token to $Q$ is already buffered in $\mathcal{P}$.  
To maintain the size of the queue, when we pop the top token, we only require to probe $\mathcal{I}$ with the query token corresponding to the popped token because the most similar tokens to the rest of the query tokens still exist in $\mathcal{P}$.    
Note that when a non-zero $\alpha$ threshold is provided, we stop probing $\mathcal{I}$ when the first  token with a similarity smaller than $\alpha$ is retrieved. 

In addition, we build an\eat{classic} inverted index $\mathcal{I}_s: \mathbb{D} \to \mathbb{P}(\mathcal{L})$ that maps each token in $\mathbb{D}$ to the sets in $\mathcal{L}$ that contain the token. 
Note that the stream of sets from $\mathcal{I}_s$ contains duplicate sets, however, 
since $\mathcal{I}_e$ returns tokens in the descending order of similarity to $Q$, the sets that are retrieved from $\mathcal{I}_s$, if have not been seen before, arrive at the descending order of UB-Filter.   
The space complexity of $\mathcal{I}_e$ is in the size of the vocabulary, $|\mathbb{D}|$,  plus the size of the query,  $|\mathbb{D}|+|Q|$. 
The space complexity of $\mathcal{I}_s$ is in the size of the vocabulary times the number of sets in $\mathcal{L}$, $|\mathbb{D}|\times|\mathcal{L}|$.  
Next, we describe how the tokens retrieved by  $\mathcal{I}_e$ help  with refining the bounds of all sets. 
}

\section{Refinement: Advanced Filters}
\label{sec:refinement}

If a set obtained from the index is not pruned using the UB-Filter, it is added to the candidate collection. 
We  \camera{introduce two advanced bounds, $iUB(C_i)$ and $iLB(C_i)$, and} describe how we incrementally update \camera{these} bounds.\eat{, namely  $iUB(C_i)$ and $iLB(C_i)$.} \fnrev{\camera{$iUB$ and $iLB$} are based on the  partial bipartite greedy matching  of sets. 
For an element in a candidate set, \eat{$\mathcal{I}_s$}\camera{$\mathcal{I}_e$} may provide multiple matching elements (all edges in a bipartite graph). However, we only consider valid edges, i.e., those matching  unmatched elements.} 

\eat{\begin{figure}[htbp]
\centering
    \includegraphics[width=0.7\linewidth]{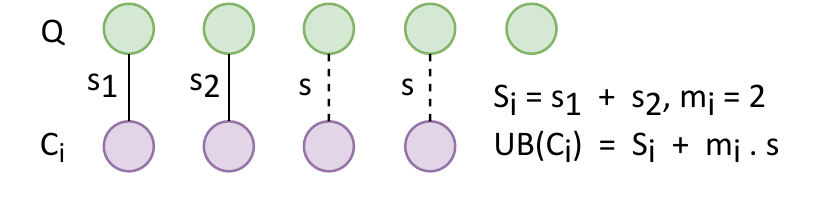}
    \caption{Updating an upper bound upon the retrieval of an element.}
\label{fig:bounds}
\end{figure}}

\noindent{\bf iLB:} \fnrev{In $iLB$, we assume all edges that are not a part of a partial greedy matching have  similarity zero. 
\begin{lemma}\label{lemma:ilb} Given\eat{a} candidate set $C$ and\eat{a} query $Q$, the sum of weights of any subset of edges in a  bipartite greedy matching of $Q$ and $C$ is a lower-bound for  \camera{$\mathcal{SO}(C)$.}\eat{of $Q$ and $C$.}  
\end{lemma} 
\begin{proof} \eat{Suppose $M^\prime: Q\rightarrow C$  is the result of greedy matching  on the bipartite graph  $G^\prime=(V,E^\prime)$. 
Let  $s^\prime$ be the score of $M^\prime$.} 
\camera{Let $G^\prime=(V,E^\prime)$ be the bipartite graph of the greedy matching $M^\prime: Q\rightarrow C$ with score $s^\prime$.} Suppose $M^{\prime\prime}: Q\rightarrow C$ is a matching \eat{on}\camera{with} $G^{\prime\prime}=(V,E^{\prime\prime})$, where $E^{\prime\prime}\subseteq E^\prime$; 
let $s^{\prime\prime}$ be the score of $M^{\prime\prime}$. 
Because $E^{\prime\prime}\subseteq E^\prime$, we have $s^{\prime\prime}\leq s^\prime$. Based on Lemma~\ref{lemma:lb},  $s^\prime$ is  \eat{$LB(C)$}\camera{a lower bound on $\mathcal{SO}(C)$}, 
therefore, \camera{$s^{\prime\prime} \leq s^\prime \leq \mathcal{SO}(C)$.}\eat{we have the weight of any partial greedy matching as $LB(C)$.}  
\end{proof}
}

\fnrev{Suppose $iLB_l(C)$ is the current  score of the partial greedy matching of $C$.  Upon reading an unmatched element of $C$ with similarity $s_{l+1}$ to an unmatched query element, using Lemma~\ref{lemma:ilb},  
the lower-bound is updated to $iLB_{l+1}(C) = iLB_{l}(C) + s_{l+1}$. Since we obtain the edges of partial greedy matching from  $\mathcal{I}_s$, 
we always consider the partial matching with the highest score, thus, computing the largest $iLB$.} 
It is straightforward to verify if an $s_l$ must be included in the matching by keeping track of the set of elements that have been matched so far  between $Q$ and each $C$. 
\fnrev{Updating the lower bound of a candidate set 
may result in updating the top-$k$ list and $\theta_{lb}$. Since $\theta_{lb}$ is always increasing (otherwise, we would not  update), updating $\theta_{lb}$ results in pruning more sets based on their upper-bounds.} 

Since set $C$ may contain identical elements to  query elements, i.e., $C\cap Q\neq \emptyset$ and 
the index returns elements in decreasing order, the lower-bound of $C$ is reduced to the number of  its overlapping elements with the query, $|Q\cap C|$, plus the greedy matching score of the remaining elements of $Q$ and $C$.  
Because lower-bounds update $\theta_{lb}$,  and a tighter $\theta_{lb}$ improves the pruning power of our technique, we choose to initialize the lower-bound of a set to its vanilla overlap (number of identical elements). To do so,  the algorithm always includes a query element itself in the  result of probing $\mathcal{I}_e$ for the first time. 
With this strategy, we deal with out-of-vocabulary elements. For example, if $\tsim$ is the cosine similarity of the embedding vectors of tokens and some tokens are not covered in the embedding corpus, we still consider them in the semantic overlap calculation if the query contains the same tokens.  

\noindent{\bf iUB-Filter:} \fnrev{\eat{In}$iUB$ \eat{we assume}\camera{assumes the largest possible similarity for} any edge that is not a part of a partial greedy matching\eat{have the largest similarity}, i.e., the smallest similarity seen so far 
from  \eat{$\mathcal{I}_s$}\camera{$\mathcal{I}_e$}.} 
\fnrev{\begin{lemma}\label{lemma:iub} 
Given a partial greedy matching \camera{of cardinality $l$} of query $Q$ and set $C_i$ with score $S_i$\eat{, upon}\camera{. Upon} obtaining  $(q_m,t,s)$ \camera{from $\mathcal{I}_e$}, where  $q_m\in Q$, element $t$ belongs to any set in $\mathcal{L}$, and $s$ is the next largest element pair similarity,  
we have the upper bound \camera{$iUB(C_i) = S_i + min(|Q|-l, |C_i|-l)\cdot s \geq \mathcal{SO}(C_i)$}. 
\end{lemma}
\begin{proof} Since $l$ elements of $C_i$ have already been matched to  elements of  $Q$, set $C_i$ has maximum \camera{$m_i = min(|Q|-l, |C_i|-l)$} remaining elements to be \camera{matched.}\eat{with $Q$.}  Upon reading an element $t$,  \camera{with the next largest similarity, namely $s$,} 
regardless of which set  $t$ belongs to, 
by  definition,\eat{we have that} any unseen element of\eat{any} $C_i$ has similarity smaller than or equal to $s$.  There are two scenarios: either element $t\in C_i$  or $t\notin C_i$. 
If $t\in C_i$ is already matched with some element in the query, we discard the element and know that any unmatched element has a similarity \eat{less}\camera{no larger} than $s$ to a query element. Otherwise, the upper-bound of any unmatched element in $C_i$  after observing element $t$ can be tightened to $s$ and we have  $iUB(C_i) = S_i + m_i\cdot s$. 
\end{proof}}
Updating $iUB(C_i)$ can result in pruning the set\eat{,} if $S_i + m_i~.~s\leq\theta_{lb}$. 
A naive way is to update the upper-bound of all sets, 
whenever a  new element is retrieved from \camera{$\mathcal{I}_e$}. However, this results in an excessive number of small updates many of which will not sufficiently decrease the upper-bound to prune the updated set. 
To solve this issue, we propose a technique that groups sets into buckets by their number of unseen elements ($m$). Only sets that contain a newly retrieved element require an update and are moved to bucket $m-1$. All other sets need not be updated, but still are pruned as soon as their upper-bound falls below $\theta_{lb}$.  

\algo bucketizes sets 
into $m$ buckets 
$B_m =\{(C_i, S_i) \vert m_i=m\}$. 
Upon the arrival of a new element with similarity $s$, any  set $C_i$ in $B_m$ and a sum of matched elements $S_i$ should be pruned if its updated upper-bound is smaller than $\theta_{lb}$, i.e. $S_i +  m\cdot s \leq \theta_{lb}$. We conclude that if $S_i\leq\theta_{lb} - (m\cdot s)$, we can safely prune $C_i$. Since all sets in a bucket have the same number of remaining elements and $s$ is fixed for all sets, the right-hand side of this pruning inequality is the same for all sets in a bucket. 

We maintain the pairs $(C_i,S_i)$ in a bucket ordered by ascending $S_i$ values. Upon the arrival of an element with similarity $s$, 
we scan the ordered list of sets in each  bucket. 
If a pair satisfies the condition $S_i\leq\theta_{lb} - (m\cdot s)$, we prune the set. As soon as we find a set $S^\prime$ that does not satisfy the condition, we can conclude that the remaining sets  do not satisfy the condition and cannot be pruned, because the remaining sets must have $S_i$ values larger than $S^\prime$. 
Now, suppose a set $C_i$ in bucket $B_m$ contains the newly arrived element: We first remove the pair $(C_i,S_i)$ from its bucket, update $S_i$, and insert the pair into $B_{m-1}$. Restricting updates to the sets that contain an element and  using the element similarity to prune some sets 
saves us many updates at each iteration. Our experiments confirm that maintaining buckets does not incur a large overhead.  

Since the similarity of identical elements is  one even  for identical out-of-vocabulary elements, like $iLB$, \algo  initializes the  upper-bound of a newly obtained set $C_i$ and its  $S_i$ to its vanilla overlap.

\begin{figure}[tbp]
\centering
    \includegraphics[width=\linewidth]{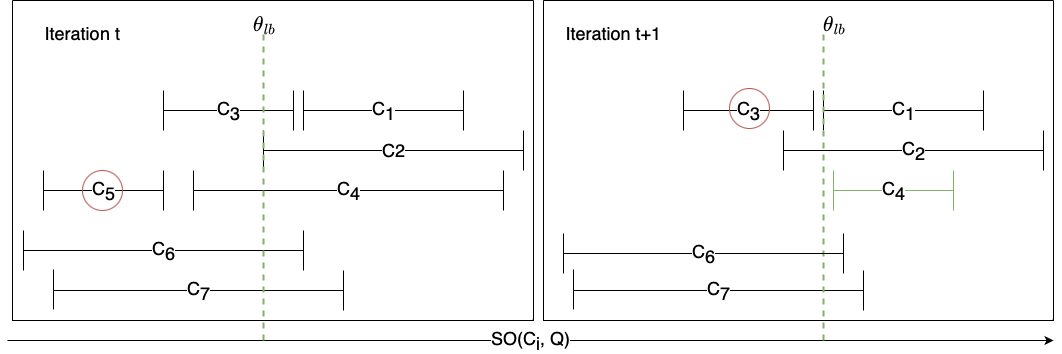}
    \caption{\fnrev{Pruning sets with upper and lower bounds for $k=2$. Each set is represented with  its LB and UB.} \pranayrev{The red circles indicate the pruned sets and the green lines indicate the recent updates to LB and UB.}}
\label{fig:intervals}
\end{figure}

\begin{example}  \pranayrev{Consider sets $C_1, \ldots, C_7$ in Fig.~\ref{fig:intervals}. Suppose we are searching for the top-$2$ sets with the highest semantic overlap with a query. Each set is represented with an  interval of its lower-bound and upper-bound with respect to the query. The  value of $\theta_{lb}$ is initially calculated based on the lower-bounds of $C_1$ and $C_2$, since they have the top-$2$ lower bounds among all sets. This means at the beginning of iteration $t$, set $C_5$ is pruned because $UB(C_5)<\theta_{lb}$.   The remaining sets stay in the candidate collection because they all have lower-bounds smaller than  $\theta_{lb}$ and upper-bounds that are greater than $\theta_{lb}$. 
Suppose, at iteration $t + 1$, by reading an element of $C_4$ from the token stream,  $LB(C_4)$ and $UB(C_4)$ are  updated (the right side of Fig.~\ref{fig:intervals}). Now, $C_4$ is  the set with the highest lower bound and as a result, the value of $\theta_{lb}$ is updated to $LB(C_1)$. This allows us to safely prune $C_3$ because   $UB(C_3)<\theta_{lb}$. 
At the end of the refinement phase, the algorithm passes remaining sets \fnrev{($C_1$, $C_2$, $C_4$, $C_6$, and $C_7$)} to the post-processing phase. \label{ref:intervals1}} 
\end{example}

\eat{
\begin{figure}[htbp]
\centering
    \includegraphics[width=0.7\linewidth]{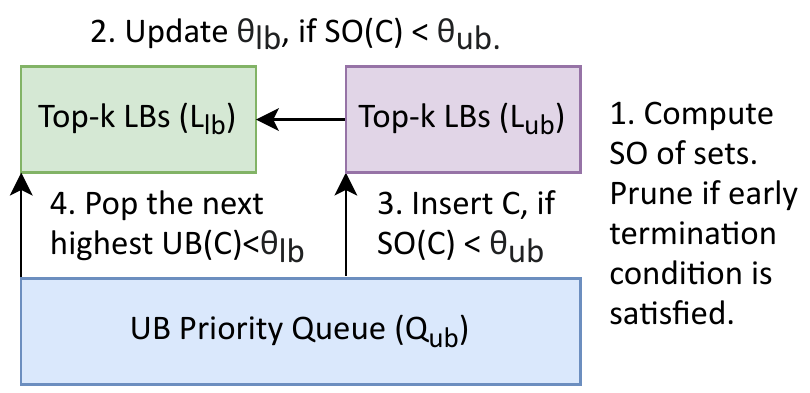}
    \caption{Post-processing steps.}
\label{fig:postprocdiag}
\end{figure}
}
\vspace{-2em}
\section{Post--Processing Phase}
\label{sec:postproc}

All candidate sets that have not been pruned during the refinement need to be verified.   
Some sets 
may end up in this phase due to \camera{their large cardinality} even \eat{if}\camera{though} the  similarity of their elements 
is relatively low. \algo applies filters during post-processing to minimize the number of exact match calculations as well as the time to complete the calculation.

\noindent{\bf No-EM:} \fnrev{
\eat{We define $\theta_{ub}$ to be the the $k$-th smallest $UB(C)$, such that $C\in\mathcal{R}$.}
\camera{Let $\mathcal{U}$ be the candidate sets that have not been pruned, and $\theta_{ub}$  be the $k$-th largest $UB(C)$\eat{,} for $C\in\mathcal{U}$.}  
\begin{lemma}\label{lemma:noem} \eat{A set $C$ with $\theta_{ub}\leq LB(C)$ is guaranteed to be in a  top-$k$ result.}\camera{A set $C \in \mathcal{U}$ with $\theta_{ub}\leq LB(C)$ is guaranteed to be in a top-$k$ result $\omega$.} 
\end{lemma}
\begin{proof} \eat{Upon maximum  graph matching calculation for set $C$, we have $\mathcal{SO}(C)=LB(C)=UB(C)$. Therefore, 
a maximum graph matching score never increases  $\theta_{ub}$\eat{,} because $\mathcal{SO}(C)\leq UB(C)$. If a set $C$ in  $\mathcal{L}_{ub}$ satisfies $\theta_{ub}\leq LB(C)$, by $LB(C)\leq UB(C)$ and $\theta^*_{k}\leq \theta_{ub}$, 
it is guaranteed that $\theta^*_{k}\leq\mathcal{SO}(C,Q)$ and $C$ is in the  top-$k$ result.} 
\camera{Computing the maximum graph matching for a $C \in \mathcal{U}$ never increases  $\theta_{ub}$, because $\mathcal{SO}(C)\leq UB(C)$. If a set $C \in \mathcal{U}$ satisfies $\theta_{ub}\leq LB(C)$, by $LB(C)\leq UB(C)$ and $\theta^*_{k}\leq \theta_{ub}$,  
it is guaranteed that $\theta^*_{k}\leq\mathcal{SO}(C)$ and $C$ is in the  top-$k$ result.}
\end{proof}
}
\fnrev{This Lemma allows us to skip  the exact semantic overlap calculation of some sets. Due to ties at distance $\theta_k^*$ from the query there can be multiple solutions for the top-$k$ search problem. All solutions share the same value for $\theta_k^*$.}


\begin{algorithm}[t]
\caption{\algo(\sopost)}
\begin{footnotesize}
\begin{algorithmic}[1]
    \Require $Q$: a query set, $k$: search parameter, $\mathcal{U}$: unpruned sets, $k$: search parameter, $\topklb$: top-$k$ LB list
    \Ensure $\topkub$: top-$k$ results
	\State $\pqub\gets{\tt init\_pq\_UB(\mathcal{U})}$ // priority queue on unpruned sets
	\State $\topkub\gets{\tt init\_topk\_UB(\mathcal{U})}$ // top-$k$ UB list on unpruned sets
    \While{$\neg\topkub.{\tt all\_checked()}$}
        \State $C\leftarrow\topkub.{\tt select\_next\_unchecked}()$
        \If{$LB(C)\geq UB(\topkub.{\tt top}())$}
            \State C.${\tt checked}\leftarrow$True
            \State continue
        \EndIf
        \State SO(C)$\leftarrow${\tt compute\_SO\_early\_termination}(C,Q)
        \State LB(C),UB(C)$\leftarrow$SO(C), C.${\tt checked}\leftarrow$True 
        \If{$SO(C)<\topkub.{\tt bottom}()$}  $\topkub.{\tt remove}(C)$
            \If{$SO(C)\geq\topklb.{\tt bottom}()$} $\pqub.{\tt add}(C)$ 
            \EndIf
        \EndIf
        \If{$SO(C)\geq\topklb.{\tt bottom}()$} $\topklb.{\tt update}(C)$
        \EndIf
        \While{$\topkub.{\tt len}() < k$ and $\neg\pqub.{\tt empty()}$}
        \State $C\gets\pqub.{\tt pop}()$
	    \If{$\topklb.{\tt bottom}()<UB(C)$} $\topkub.{\tt add}(C)$  
        \EndIf
		\EndWhile
	\EndWhile
	\Return $\topkub$
\end{algorithmic}
\end{footnotesize}
\label{alg:postproc}
\end{algorithm}

In this phase, \algo maintains three data structures: 1) an ordered list of sets with top-$k$ lower-bounds ($\mathcal{L}_{lb}$), 2)   an ordered list of sets with top-$k$ upper-bounds ($\mathcal{L}_{ub}$), and 3) a priority queue of sets ordered by upper-bounds ($\pqub$). Maintaining $\mathcal{L}_{lb}$ and $\mathcal{L}_{ub}$ allows us to have fast access to $\theta_{lb}$ and $\theta_{ub}$, respectively.  
Based on Lemma~\ref{lemma:noem}, the algorithm should only compute the bipartite matching of sets with $UB(C)\geq\theta_{ub}$. 
As such, \algo prioritizes the exact graph-matching  calculation of sets with high upper-bounds.  Intuitively, sets with high upper-bounds have the potential for high semantic overlaps. 
To speed up,  all sets in $\mathcal{L}_{ub}$ are queued and evaluated in parallel in the background. Upon the completion of the exact match of a set $C$, we update  $LB(C)=UB(C)=\mathcal{SO}(C)$. This has two effects. 
First, the update of $UB(C)=\mathcal{SO}(C)$ may cause the $\theta_{ub}$ to be 
decreased. As a result, if $\mathcal{SO}(C)\geq\theta_{ub}$, the set remains in $\mathcal{L}_{ub}$. 
If $\mathcal{SO}(C)<\theta_{ub}$, we add the set to $\pqub$, because the algorithm may realize later that $\mathcal{SO}(C)$ was higher than the sets that are currently in $\mathcal{L}_{ub}$, whose semantic overlaps are not calculated yet. 
Inserting a set into $\pqub$ results in $\mathcal{L}_{ub}$ having $k-1$ sets. Probing $\pqub$ provides the next set with the $k$-th largest  upper-bound to be added to $\mathcal{L}_{ub}$. 

Second, the update of $LB(C)$ may cause the $\theta_{lb}$ to increase and potentially prune some sets. The algorithm takes a lazy approach and considers the sets in $\mathcal{L}_{ub}$ for such pruning until a set $C$ with $UB(C)\leq\theta_{lb}$ is obtained.  
The post-processing phase terminates when all sets in the $\mathcal{L}_{ub}$ satisfy the condition $\theta_{ub}\leq LB(C)$ and the list is returned as the final search result. 

\noindent{\bf EM-Early-Terminated-Filter:}  Despite pruning sets extensively, the exact matching calculation remains expensive.  
Consider bipartite graph $G(V,E)$ built on $C$ and $Q$, where $V={Q} \uplus {C}$ and the weight of an edge between elements $q_i\in Q$ and $c_j\in C$ is \camera{$w(q_i,c_j)=\tsim(q_i,c_j)$.} The weight of an optional one-to-one matching \eat{$M\subset E$}\camera{$M\subseteq E$} is $w(M) = \sum_{(q_i,c_j)\in M}w(q_i,c_j)$.   
A matching $M$ is called perfect if for every $v\in V$, there is some $e\in M$ which is incident on $v$. The Hungarian algorithm~\cite{Munkres57} considers a node labeling to be a function $l: V \rightarrow \mathbb{R}$. A feasible labeling satisfies $l(q_i)+l(c_j)\geq w(q_i, c_j), \forall q_i\in Q, c_j\in C$. An equality subgraph is a subgraph $G_l = (V,E_l) \subset G = (V,E)$, fixed on a labeling $l$, such that $E_l = \{(q_i,c_j)\in E: l(q_i)+l(c_j)=w(q_i, c_j)\}$. The Hungarian algorithm considers a valid labeling function $l$ and maintains a matching $M$ and a graph $G_l$. It starts with $M=\emptyset$. At each iteration, the algorithm either augments $M$ or improves the labeling $l\rightarrow l^\prime$\eat{,} until  
$M$ becomes a perfect matching on $G_l$. Let $M^\prime$ be a perfect matching in $G$ (not necessarily in $G_l$). It is a well-known result that if $l$ is feasible and $M$ is a perfect matching in $G_l$, then $M$ is \camera{a} maximum weight matching~\cite{Munkres57}. 
From the proof of the Kuhn-Munkres theorem, we have that $w(M^\prime)\leq w(M)$.  
This  theorem also shows 
that for any feasible labeling $l$ and any matching $M$ we have $w(M)\leq  \sum_{v\in V}l(v)$~\cite{Munkres57}. 
\eat{$w(M^\prime)\leq\sum_{v\in V}l(v)$.}

\begin{lemma}\label{lemma:hungarian} A set $C$ can be safely pruned during bipartite graph matching with $Q$\eat{,} if the sum of labels assigned by the Hungarian algorithm~\cite{Munkres57} is smaller than $\theta_{lb}$. 
\end{lemma}
\begin{proof} 
Suppose a valid node labeling function $l$ and an equality subgraph $G_l$.  The perfect matching $M$ on $G_l$ created by the Hungarian algorithm  is indeed the maximum matching and its weight $w(M)$ is the semantic overlap of $Q$ and $C$.    
Now, let $M^\prime$ be some perfect matching in $G$.  
Since from the Kuhn-Munkres theorem we know that $w(M^\prime)\leq w(M)$ and for the labeling $l$ and  matching $M$, $w(M)  \leq  \sum_{v\in V}l(v)$, we have $w(M^\prime)\leq\sum_{v\in V}l(v)$. 
Therefore, 
the sum of node labels is an upper bound for $\mathcal{SO}(Q,C)\leq\sum_{v\in V}l(v)$. This upper bound can be computed on the fly during graph matching. A set $C$ can be pruned as soon as $\sum_{v\in V}l(v)$ exceeds $\theta_{lb}$. 
\end{proof}
\fnrev{By computing this upper-bound during the process of graph matching for a set $C$, as soon as  $UB(C)<\theta_{lb}$, the process can be terminated and $C$ can be safely removed.}
The early termination is particularly important for  large sets when the matching of a large set is executed in parallel and a global $\theta_{lb}$ is updated as the processing of other sets is completed.  

\begin{example}\label{ref:intervals2} Suppose sets $D_1$, $\ldots$, $D_6$ of Fig.~\ref{fig:postintervals} are in the post-processing phase. Suppose we are searching for the top-$3$ sets. Each set is represented with an  interval of its lower-bound and upper-bound with respect to the query. The post-processing starts by $\mathcal{L}_{ub} = [D_2, D_1, D_6]$. From $LB(D_2)>\theta_{ub}$, we know we do not need to compute the exact match of $D_2$, because we know the true value of $\theta_k$ ($\theta^*_{k}$) is never greater than $\theta_{ub}$. Thus,  any set with semantic overlap higher than $\theta^*_{k}$ must be in the final result. The exact matching of $D_1$ and $D_6$ are calculated in parallel. Suppose $D_6$ finishes and its bounds are updated. This results in increasing $\theta_{lb}$ and decreasing $\theta_{ub}$. Thus, $D_6$ is removed from $\mathcal{L}_{ub}$ and added  to $\mathcal{Q}_{ub}$. After probing $\mathcal{Q}_{ub}$ considering $\theta_{lb}$, $D_3$ is added to $\mathcal{L}_{ub}$. Suppose using the early termination rule, $UB(D_3)$ is updated to a value lower than $\theta_{lb}$. The algorithm immediately stops the matching of $D_3$, and eliminates it from $\mathcal{L}_{ub}$. At this point, $D_6$ is added back to  $\mathcal{L}_{ub}$ and $\theta_{ub}=SO(D_6)$. 
The algorithm continues to compute the next set in $\mathcal{L}_{ub}$ ($D_1$).   
\end{example}

\begin{figure}[tbp]
\centering
    \includegraphics[width=\linewidth]{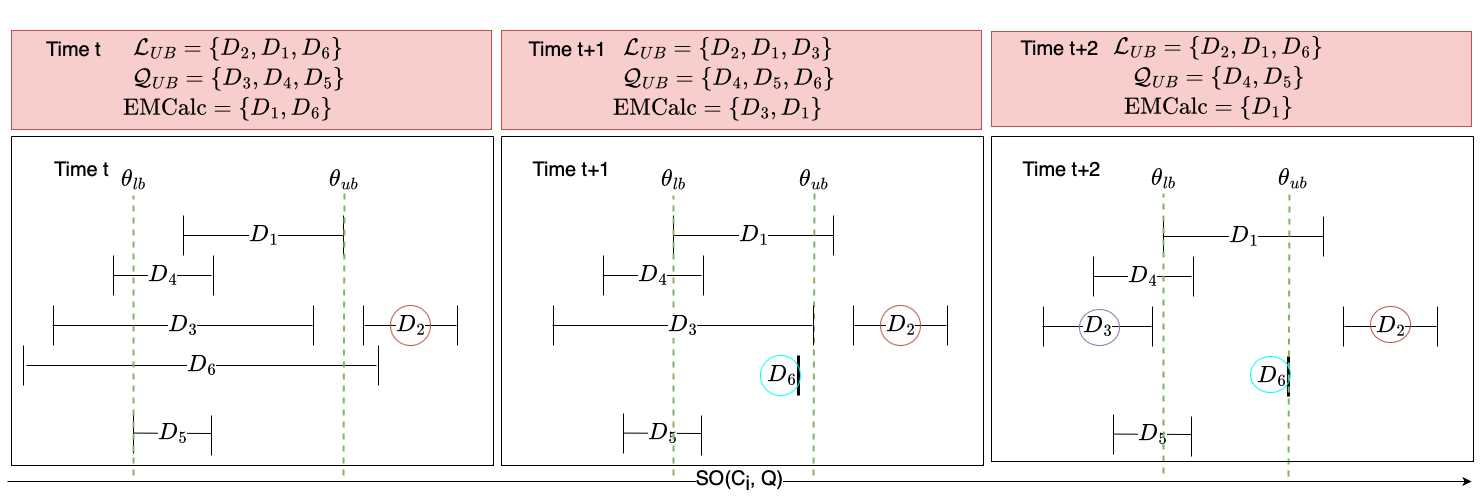}
    \caption{\pranayrev{Post-processing: pruning sets with upper and lower-bounds and exact semantic overlap, for $k=3$. Each set is represented with its LB and UB. Here the red circle indicates NoEM Filter, the blue circle indicates EM, and the purple circle indicates the EM-ETM Filter.}}
\label{fig:postintervals}
\end{figure}

To scale out search, we randomly partition the repository  and run \algo on  partitions in parallel. To improve the pruning power, all partitions share a global $\theta_{lb}$ that is the maximum of the $\theta_{lb}$. 

\eat{
\section{Parallel \algo} 
\label{sec:parallel}

A simple strategy to scale-out search is to randomly partition the repository  and 
distribute the top-$k$ search algorithm 
across the partitions and merge the local top-$k$ result of partitions to obtain the global top-$k$ result. 
We use a global $\theta_{lb}$ that is the maximum of the $\theta_{lb}$ of all partitions and is updated every time a local $\theta_{lb}$  is updated. 
Compared to the sequential version of \algo, the parallel version looks at more sets per time unit, therefore, even though the local $\theta_{lb}$ may never converge to $\theta_k^*$, because a global  $\theta_{lb}$ is used for pruning, we get high pruning power in addition to the computation speedup. Balancing the partitions to get the maximum benefit of parallelization requires a more advanced load-balancing partitioning solution and  warrants future research work. 
}

\eat{
\begin{figure*}[!ht]
    \begin{subfigure}[t]{0.24\linewidth}
        	\centering
        	\includegraphics[width =\textwidth]{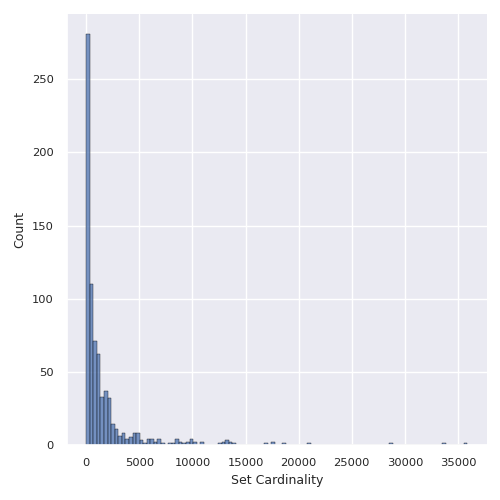} 
        	\caption{Set Cardinality  of Open Data}
            \label{fig:setdistod}
    \end{subfigure}
    \hfill
    \begin{subfigure}[t]{0.24\linewidth}
        	\centering
        	\includegraphics[width =\textwidth]{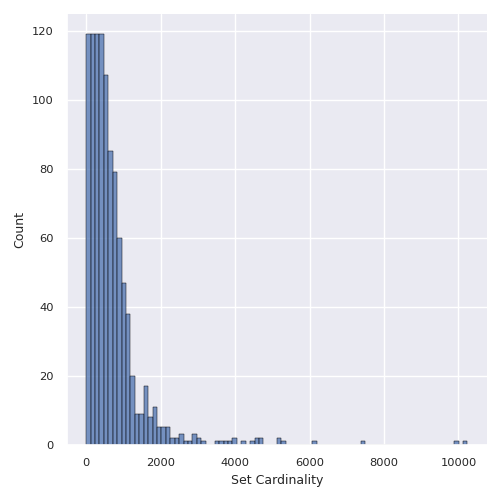} 
        	\caption{Set Cardinality  of WDC}
            \label{fig:setdistwdc}
    \end{subfigure}
    \hfill
    \begin{subfigure}[t]{0.24\linewidth}
        	\centering
        	\includegraphics[width =\textwidth]{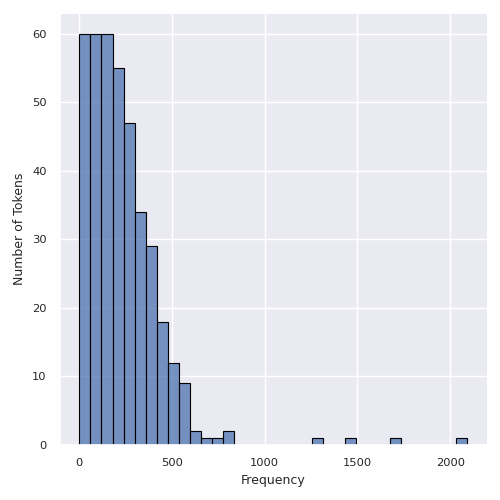} 
        	\caption{Tokens of Open Data}
            \label{fig:tokendistod}
    \end{subfigure}
    \hfill
    \begin{subfigure}[t]{0.24\linewidth}
        	\centering
        	\includegraphics[width =\textwidth]{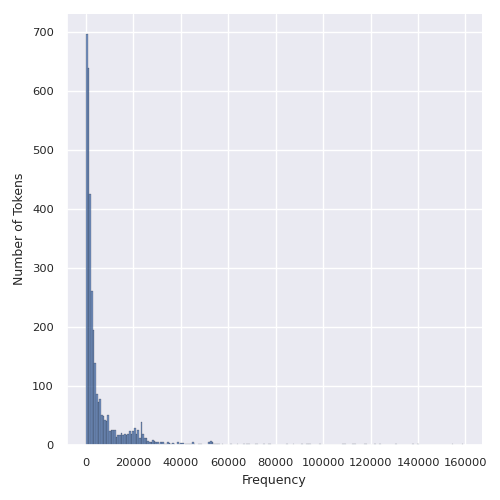} 
        	\caption{Tokens of WDC}
            \label{fig:tokendistwdc}
    \end{subfigure}
    \caption{Data Characteristics.}
\end{figure*}
}

\section{\fnrev{Analysis of \algo}}
\label{sec:analysis}

\subsection{\fnrev{Correctness}}
\label{sec:correctness}

\fnrev{
Based on Def.~\ref{def:semoverlap}, a set $C$ with at least one element with similarity greater than $\alpha$ has a non-zero semantic overlap with $Q$. 
The token stream retrieves all set elements with similarity greater than $\alpha$, and for each  element, the inverted index  returns all sets containing that element.  As a result, \eat{no set with  non-zero  semantic overlap is discarded.}\camera{all sets with non-zero semantic overlap are considered.}  
\eat{As shown in Algorithm~\ref{alg:refinement} and~\ref{alg:postproc},}
\algo can prune a set $C$ under three circumstances: 
(1) \nikolausrev{When} $C$ appears for the first time due to the similarity \camera{$s$} of one of its elements $c_j\in C$ to \eat{a}\camera{some} query element $q_i\in Q$. \eat{In fact, $\tsim(q_i,c_j)), q_i\in Q$ is}\camera{Because of the non-increasing similarity order in $\mathcal{I}_e$, $s = \max\{\tsim(q_i,c_j)\}, c_j\in C, q_i\in Q$.}  
Therefore, based on Lemma~\ref{lemma:ub} \camera{, the upper bound $UB(C)$ can be computed,} and $C$ is pruned if $UB(C)\leq\theta_{lb}$. 
Note that considering $\theta_{lb}$ for pruning does not create false negatives because,  as shown in Lemma~\ref{lemma:thetalb}, $\theta_{lb}\leq\theta_k\leq\theta^*_k$. 
(2) During refinement by iUB filter: Lemma~\ref{lemma:iub} proves an upper bound based on the similarity of elements  of arbitrary sets\nikolausrev{, thus any set pruned by $\theta_{lb}$ cannot be a false negative}. 
\eat{Again, by definition, any set pruned by $\theta_{lb}$ cannot be a false negative. }
(3) During post-processing: Lemma~\ref{lemma:hungarian} proves that whenever the sum of node labels of the Hungarian algorithm becomes smaller than $\theta_{lb}$, the set can be safely pruned. The algorithm continues to prune based on the UB-Filter in this phase as $\theta_{lb}$ is being improved by the exact match calculation.  
If a set is not pruned under the above conditions until the algorithm terminates, it remains in \eat{the} $\mathcal{L}_{ub}$ and is in fact in \camera{the} top-$k$ \camera{result}.}  

\subsection{\fnrev{Time and Space Complexity}}
\label{sec:complexity}

\fnrev{
\camera{The exact bipartite graph matching can be computed in $\mathcal{O}(n^3)$ time~\cite{Zvi,FredmanT87}\footnote{\camera{For graphs with a particular structure, a lower complexity may be achieved with Dijkstra's algorithm and Fibonacci heaps~\cite{AgarwalS14}.}}.} Although the worst-case time complexity of \algo is $\mathcal{O}(m\cdot n^3)$, where $n$ is the \camera{maximum} set cardinality \eat{of}\camera{and} $m$ is the number of sets in $\mathcal{L}$, \camera{thanks to} our filters \eat{ensure that we make}\camera{we require} far fewer \camera{than} $m$ comparisons, thereby reducing the overall runtime  by orders of magnitude in practice. Moreover, the EM-early-terminated filter can  reduce the number of iterations of the Hungarian algorithm.} 

\eat{
\subsection{Time Complexity}
\pranayrev{\algo work in two phases: Refinement and Post-processing. In the refinement phase, we process sets as they are returned by the token stream, $\mathcal{L}_e$ and update the bounds for each set. Let $Q$ be the given query, the size of the token stream would be $Q \cdot S_q$, where $S_q$ is the number of tokens that have similarity greater than $\alpha$ with the token query. Next we process, i.e, update the bounds, for all sets returned by the inverted index, $\mathcal{I}_e$, which could be in the worst case all the sets in the repository, i.e $L$. The lower bound update, \textbf{iLB} is constant time, and the upper bound update, uses a bucketization technique, hence the update in the worst case is the time taken to look up a set from a bucket. Thus the time complexity for the refinement phase is $\mathcal{O}\left(\left(\lvert Q \rvert \cdot S_q\right) \cdot \left(\lvert L \rvert \cdot \left(\log B \right)\right)\right)$ where $B$ is the size of the maximum bucket. In the post-processing phase, we have the priority queue of all possible candidate sets, $\mathcal{Q}_{UB}$, for which exact matching needs to be computed. We apply filters to further reduce the number of candidate sets. In the worst case we would have to compute the matching for all candidate sets; hence the time complexity for post-processing phase would be $\mathcal{O}\left(\lvert \mathcal{Q}_{UB} \rvert \cdot \left( \lvert C \rvert + \lvert Q \rvert \right)^3 \right)$.}
}

\eat{\fnrev{The top-$k$ LB list (top-$k$ UB list) that keeps track of the ids and set similarity scores  has size  $k.(sizeof(Int)+sizeof(Double))$.
The inverted index that maps integer token ids to a list of integer set ids has size  $\sum_{t_i\in\mathcal{L}}(f_i+1).sizeof(Int)$, where $f_i$ is the frequency of elements in $t_i\in\mathbb{D}$. 
We estimate the size of the token stream with $|Q|.S_q.(2.sizeof(Int) + sizeof(Double))$, where $S_q$ is the number of elements in the vocabulary that has a similarity greater than $\alpha$. The size of the upper-bound priority queue $\pqub$  is $|Candiates\setminus Pruned|.(sizeof(Int)+sizeof(Double))$.}}

\camera{Recall that $\mathbb{D}$ is \eat{the collection of all elements}\camera{the vocabulary of all distinct tokens in $\mathcal{L}$}. Each token in $Q$ can have a similarity greater than $\alpha$ with at most $\lvert \mathbb{D} \rvert$ tokens, thus the space complexity for the token stream $\mathcal{I}_e$ is $\mathcal{O}\left( \lvert \mathbb{D} \rvert \cdot \lvert Q \rvert \right)$. The inverted index $\mathcal{I}_s$ is linear in the input size: it stores $\lvert \mathbb{D} \rvert$ keys and the aggregate size of all lists is at most $D^{+} = \sum_{C \in \mathcal{L}}\lvert C \rvert$. With $\lvert \mathbb{D} \rvert \leq D^{+}$ and an average set size of $\Bar{C}$, the space complexity of $\mathcal{I}_s$ is $\mathcal{O}\left(\mathcal{L} \cdot \Bar{C}\right)$.\eat{it stores $\lvert \mathbb{D} \rvert$ keys and the size of each list is at most $\mathcal{L}$. 
Therefore, the overall space complexity is $\mathcal{O}\left(\mathcal{L}\cdot\mathbb{D}\right)$.} 
Both the top-$k$ lists, $\mathcal{L}_{lb}$ and  $\mathcal{L}_{ub}$, store at most $k$ sets at any given time during the iteration and have a space complexity of $\mathcal{O}\left(k\right)$. The size of the priority queue $\pqub$ is $\mathcal{O}\left(\mathcal{L}\right)$. This gives us the overall space complexity of  \eat{$\mathcal{O}\left(\mathcal{L}\cdot\mathbb{D}\right)$}$\mathcal{O}\left(\lvert \mathbb{D} \rvert \cdot \lvert Q \rvert + \mathcal{L} \cdot \Bar{C}\right)$.}

\eat{
We use the Faiss index~\cite{JDH17} that enables efficient similarity search on dense vectors. 
For the inverted index, we map integer token IDs to a list of integer set IDs. The size of this index  is simply the $\sum_{t_i\in U}  (f_i+1).sizeof(Int)$, where $f_i$ is the frequency of token $t_i\in\mathbb{D}$ and one accounts for the key of the map. Now, we turn our attention to data structures whose memory footprint is query-dependent.} 
\eat{The token stream $\mathcal{L}_e$ stores a list of edges each  with two integer IDs and its tokens and a double value of the \tsim of tokens. We estimate the size of the token stream, $\mathcal{L}_e$, with $|Q|.S_q.(2.sizeof(Int) + sizeof(Double))$.} 
\eat{In the current implementation,  the token stream is also a cache for the similarity of query tokens and candidate tokens read from the token stream during the refinement phase. This cache helps with speeding up the similarity matrix initialization of exact match calculation in the post-processing phase. In fact, no similarity calculation is needed during post-processing since all tokens with similarity greater than $\alpha$ to a query token are already read from the token stream and exist in the cache.} \eat{Recall  $S_q$ is the number of tokens in $U$ that has a similarity greater than $\alpha$ with token query $q$ (\tsim$\neq 0$). Note that in the current implementation, $\mathcal{L}_e$ is also a cache for the similarity of the query and candidate tokens.} 

\eat{For efficiency purposes, to keep track of candidate sets we have to maintain lists of IDs of candidate sets and IDs of pruned sets with the total size of  $|Candiates|.sizeof(Int)+|Pruned|.sizeof(Int)$. 
To incrementally compute the lower-bound of a candidate set that is not  already pruned, the algorithm maintains the matched tokens of the candidate as well as the matched tokens of the query in the greedy matching of the set and the query. The total size of the lower-bound data structure is $|Candiates\setminus Pruned|.(sizeof(Int)+2.|Q|.sizeof(Int)+sizeof(Double)))$. The size of the lower-bound map is the number of unpruned sets ($|Candiates-Pruned|$), each represented with an integer ID. For each set, the \algo\eat{\algoplus} stores the integer IDs of the matched tokens of the set and the matched tokens of the query. The number of matched tokens of a candidate set is at most $|Q|$. Therefore, the algorithm stores at most $|Q|$ integer token IDs for the query and for the candidate set in the second term. Finally, the double value of the incremental lower bound of the set is stored to avoid the calculation of the matching score every time a new edge is added to the matching. 
The size of the upper-bound data structure is  $|Candiates\setminus Pruned|.(sizeof(Int)+sizeof(Double))$. 
As described in \S~\ref{sec:refinement}, we use buckets for organizing the incremental upper bounds of unpruned sets. Each entry in a bucket contains the integer ID of a set and the double value of its current upper bound. 

The size of the upper-bound priority queue $\pqub$  during the post-processing  is $|Candiates\setminus Pruned|.(sizeof(Int)+sizeof(Double))$. This queue maintains the IDs of sets as well as their upper-bound scores. 
Finally, we avoid recalculating the exact match of candidate sets when they appear at the top of the upper-bound priority queue, for each set, the algorithm considers a boolean value to track whether the exact match of the set is already calculated.}  

\begin{table}[t]
\centering 
\caption{Characteristics of datasets.}
\label{tbl:datasets}
    \begin{tabular}{|c|cccc|}
    \hline
     & \#Sets & MaxSize & AvgSize & \#UniqElems \\
    \hline
    \pranay{DBLP} & 4,246 & 514 & 178.7 & 25,159\\
    OpenData & 15,636  & 31,901 & 86.4 & 179,830\\
    \pranay{Twitter} & 27,204 & 151 & 22.6 & 72,910 \\
    WDC & 1,014,369   & 10,240 & 30.6 & 328,357\\  
    \hline
    \end{tabular}
\end{table}

\section{Experiments}
\label{sec:experiments}

We evaluate the response time, memory footprint, and pruning power of filters of \algo on four real-world datasets and test various parameter settings and query cardinalities. 
\fnrev{We compare \algo with a baseline and a state-of-the-art fuzzy set similarity search technique.}  
\fnrev{The usefulness of semantic overlap measure for search is evaluated by comparing to the results of  vanilla set overlap search}. In our experiments, we use the cosine similarity of the embedding vectors of tokens using pre-trained vectors\footnote{https://fasttext.cc/docs/en/english-vectors.html} of FastText~\cite{GraveMJB17} as the function $\tsim$.  

\eat{
\subsection{Data Sets} 
\label{sec:datasets}

We used four data sets: DBLP~\cite{Jelodar2020}, OpenData~\cite{NargesianZPM18}, Twitter~\cite{DBLP:journals/corr/abs-2004-03688}, and the public corpus of WebTables (WDC)~\cite{LehmbergRMB16}.
For OpenData and WDC, first, we extracted sets by taking the distinct values in every column of every table. We removed all numerical values, as they create casual matches that are not meaningful, as the embedding vectors of numerical values are not accurate. 
Then, we identified the vector representation of all elements of all sets in a language model (we used the pre-trained model\footnote{https://fasttext.cc/docs/en/english-vectors.html} of  FastText~\cite{GraveMJB17}), then 
we discarded all sets that have less than 70\% of their elements covered by the pre-trained model. For DBLP we filter the original dataset by only considering publications in the years $2018$, and $2019$. For Twitter, we only consider the tweets in English. Note our algorithm still handles words that appear and might not be covered by the pre-trained model.

The original OpenData repository contains 120,213 sets. After  cleaning, we acquired 15,636 sets with at least 70\% coverage in FastText; this is 13\% of the raw data.  
The original WDC repository contains $240,653,580$ sets. After  cleaning, we acquired $1,030,374$ sets; this is $0.0042\%$ of the raw data. The characteristics of the extracted sets are
shown in Table~\ref{tbl:datasets}. \pranay{We see that WDC is the largest of the datasets with almost 239× more sets compared to DBLP, while DBLP has the largest average size.}
\eat{WDC has 65× more sets than OpenData, while its average set size is much smaller (30  vs. 86).}Fig.~\ref{fig:setdistod} and~\ref{fig:setdistwdc} show the distribution of set cardinality in OpenData and WDC, respectively. Note that the Uniqelements column indicates  the number of unique elements in the repository that are covered by the FastText database. The majority of sets in WDC are small while there exist more large sets in OpenData than in WDC. The  Hungarian algorithm~\cite{Munkres57}  has cubic complexity in the cardinality of sets. This means exact match calculation for a set in WDC is usually cheaper than OpenData. 

Fig.~\ref{fig:tokendistod} and~\ref{fig:tokendistwdc} show the distribution of element frequency in OpenData and WDC, respectively.  Unlike in OpenData, there are some very frequent elements in WDC, which results in excessively large posting lists in $\mathcal{I}_s$. As a result, the number of candidate sets in WDC during the refinement is large, and  updating the bounds of the sets is often more expensive in WDC than OpenData as we will observe in the experiments.

\subsection{Benchmarks} 
\label{sec:benchmark} 

We generated one query benchmark from each data set.
Each benchmark is a set of queries (sets) selected from a cardinality 
range in order to evaluate performance on different query
cardinalities. 
For OpenData, the benchmarks are from six intervals: 20
to 250, 250 to 500, 500 to 750, 750 to 1k, 1k to 5k, and 5k to 32k. 
For WDC, we used different intervals for benchmarks:
10 to 250, 250 to 500, 500 to 750, 750 to 1k, and 1k to 11k. 
For each interval, we sample 50 and 100 sets using uniform random sampling for OpenData and WDC,  respectively. 
We report the average of results over the queries for each interval.
Sampling by interval prevents a benchmark
that has the same skewed distribution as the repository itself,
and is heavily biased to sample smaller sets. 
Large intervals are considered for high cardinality sets because of the power-law distribution of set cardinality~\cite{ZhuNPM16}. 
The last interval of WDC is selected such that the requirement of having 100 query sets from an interval is fulfilled.  
}

\subsection{Experimental Setup}
\label{sec::exp-setup}

\subsubsection{\bf{\em{Datasets}}}
\label{sec:datasets}
We use four datasets: DBLP~\cite{Jelodar2020}, OpenData~\cite{NargesianZPM18}, Twitter~\cite{DBLP:journals/corr/abs-2004-03688}, and the public corpus of WebTables (WDC)~\cite{LehmbergRMB16}. 
For DBLP, we  
\fn{consider} papers from 2018 and 2019, and for each publication, we generate a set of white-spaced words from the paper title and abstract. For each \fn{English tweet} in the  Twitter dataset, 
\fn{we generate} a set consisting of the distinct words in the tweet 
except the emojis and URLs.\eat{For} \pranay{The sets for} OpenData and WDC\eat{, first, we extracted sets by taking} \pranay{are formed by} the distinct values in every column of every table. For all datasets, we remove numerical values to avoid casual matches. 
We  further filter OpenData and WDC  by discarding all sets that have less than 70\% coverage of  pre-trained vectors.  
The characteristics of the extracted sets are
shown in Table~\ref{tbl:datasets}.  
The majority of sets in WDC and Twitter are small but they contain more sets compared to OpenData and DBLP. 
Unlike in OpenData, there are some very frequent elements in WDC, which results in excessively large posting lists in \pranay{the inverted index.} As a result, the number of candidate sets in WDC during the refinement is large, and  updating the bounds of the sets is often more expensive in WDC than OpenData. 

\subsubsection{\bf{\em{Benchmarks}}}
\label{sec:benchmark}

We generated one query benchmark, i.e., a collection of query sets, from each data set. 
The set cardinalities in WDC and OpenData are highly skewed~\cite{ZhuDNM19}.   
\fn{In order to evaluate the performance depending on the query cardinality, the benchmarks of WDC and OpenData are  collections} of query sets selected from different cardinality ranges. The ranges for OpenData are: 10 to 750, 750 to 1000, 1000 to 1500, 1500 to 2500, 2500 to 5k, and 5k to 32k; the intervals for WDC are:
10 to 250, 250 to 500, 500 to 750, 750 to 1k, and 1k to 11k. For each interval, we sample $50$ and $100$ sets using uniform random sampling for OpenData and WDC respectively. Sampling by interval prevents the benchmarks from being biased towards small sets. Large intervals are used for high cardinality sets due to the power-law distribution of the set cardinalities~\cite{ZhuNPM16}. \fn{Since DBLP and Twitter contain fewer skewed sets}, we \eat{don't}do not  create intervals and draw 100 random sets using uniform sampling. \fn{We report the average of results over the queries for each benchmark and interval. }

\subsubsection{\bf{\em{Implementation}}}
\label{sec:implementation}

The inverted index ($\mathcal{I}_s$) is computed on the fly and stored in an in-memory hash map. 
To generate the token stream we use the GPU implementation of the \eat{top-$k$}\camera{\mbox{top-$k$}} Faiss index~\cite{johnson2019billion} over high-dimensional vectors. 
We query the Faiss index in batches of 100 elements.
\pranay{The construction time of the inverted index is $1.5$, $3.0$, $1.3$, and $80$ seconds, and the construction time of Faiss index \eat{takes}is $3.6$, $9.5$, $3.8$, $12.5$ seconds for DBLP, OpenData, Twitter, and WDC respectively.}
We cache the similarity of returned vectors during the refinement phase for reuse during the  initialization  of the  similarity matrix used in  graph  matching. 
For graph matching, we use an implementation of the Hungarian algorithm~\cite{mcximing}. 
To compute the graph matching of sets in parallel during the post-processing phase, we use a C++17-compatible thread pool implementation~\cite{abs-2105-00613}. 
Unless otherwise specified, the following  parameters are used in all experiments: similarity threshold $\alpha= 0.8$, k = $10$, and partitions = $10$.

\subsubsection{\bf{\em{Baselines}}}
\label{sec:baseline}

\pranay{The baseline approach for \eat{top-$k$}\camera{\mbox{top-$k$}} semantic overlap search\eat{is to} iterates over all candidate sets \eat{sequentially} and computes their bipartite graph matchings. Candidate sets are those  
that have at least one element with similarity to\eat{the} any query elements greater than the threshold $\alpha$.} We use the token stream to get a list of candidate sets (baseline's refinement phase) and 
use a thread pool~\cite{abs-2105-00613} to parallelize the computation of 
the graph matching of all candidate sets (baseline's post-processing). Given the  sheer number of sets and high frequency of elements in WDC, computing exact graph matchings for all candidate sets\eat{becomes expensive and} \pranay{is} infeasible. 
 For example, we have $190,679$ candidate sets for a query set with cardinality $53$. 
To reduce the number of candidate sets, we activate the iUB-Filter to assist with set pruning. This is referred to as {\em Baseline+}. 
\eat{For WDC, given the large datalake and the high frequency of elements, it becomes expensive and infeasible to compute all exact graph matchings. For example, for a query set with cardinality $53$ we have $190679$ possible candidate sets. We turn on the UB-Filter to help with pruning some sets during the refinement phase. We call this {\em Baseline+}.}
\eat{
\subsubsection{\bf{\em{Default Parameters}}}
\label{sec:default}
Unless otherwise specified, the following \algo parameters are used in all experiments: similarity threshold $\alpha= 0.8$, k = $10$, and partitions = $10$.}

\subsubsection{\bf{\em{System Specifications}}}
All experiments are conducted on a machine
with 2 Intel\textsuperscript{\textregistered} Xeon Gold 5218 @ 2.30GHz (64 cores), 512 GB DDR4 memory, a Samsung\textsuperscript{\textregistered} SSD 983 DCT M.2 (2 TB), 4 GPUs - TU102 (GeForce RTX 2080 Ti).

\begin{figure*}[!ht]
    \begin{subfigure}[t]{0.24\linewidth}
        	\centering
        	\includegraphics[width =\textwidth]{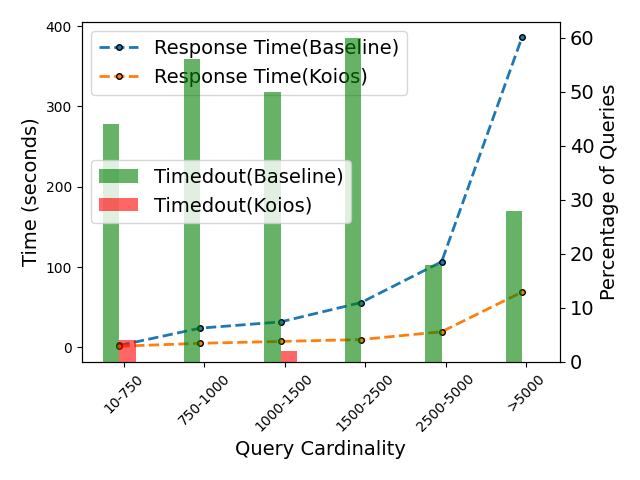}
        	\vspace{-6mm}
        	\caption{}
            \label{fig:timeresod}
    \end{subfigure}
    \hfill
    \begin{subfigure}[t]{0.24\linewidth}
        	\centering
        	\includegraphics[width =\textwidth]{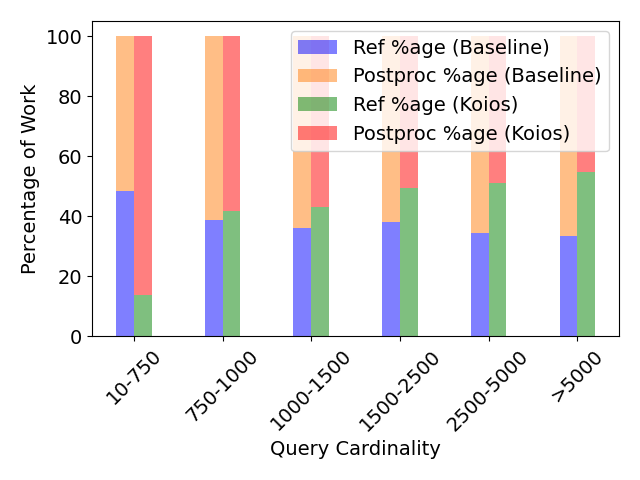} 
        	\vspace{-6mm}
        	\caption{}
            \label{fig:timepreod}
    \end{subfigure}
    \hfill
    \begin{subfigure}[t]{0.24\linewidth}
        	\centering
        	\includegraphics[width =\textwidth]{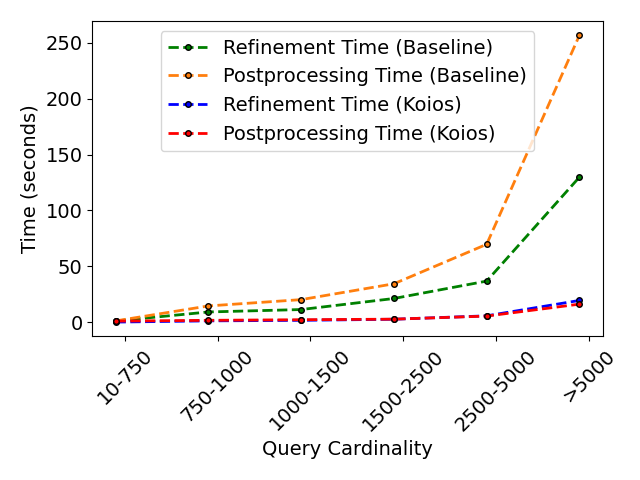} 
        	\vspace{-6mm}
        	\caption{}
            \label{fig:timebrod}
    \end{subfigure}
    \hfill
    \begin{subfigure}[t]{0.24\linewidth}
        	\centering
        	\includegraphics[width =\textwidth]{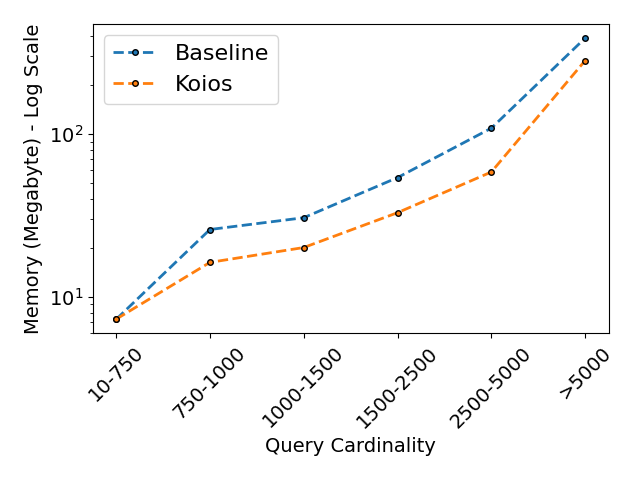} 
        	\caption{}
            \label{fig:memod}
    \end{subfigure}
    \caption{\pranayrev{\camera{OpenData results: (a) response time, (b), (c) phase breakdown, (d) memory footprint.}}}
\end{figure*}
\begin{figure*}[!ht]
\vspace{-4mm}
    \begin{subfigure}[t]{0.24\linewidth}
        	\centering
        	\includegraphics[width =\textwidth]{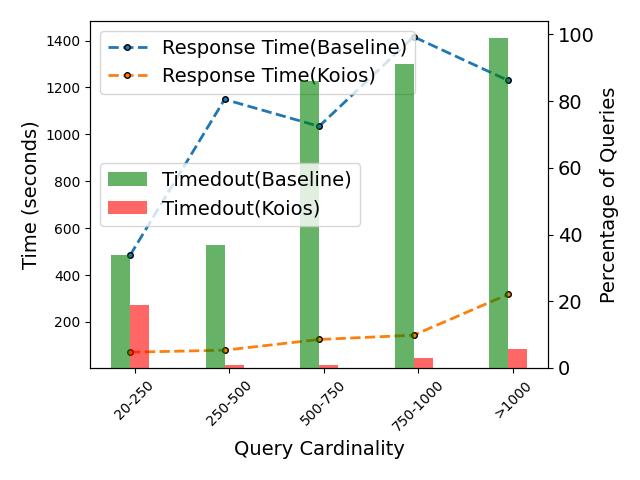} 
        	\caption{}
            \label{fig:timereswdc}
    \end{subfigure}
    \hfill
    \begin{subfigure}[t]{0.24\linewidth}
        	\centering
        	\includegraphics[width =\textwidth]{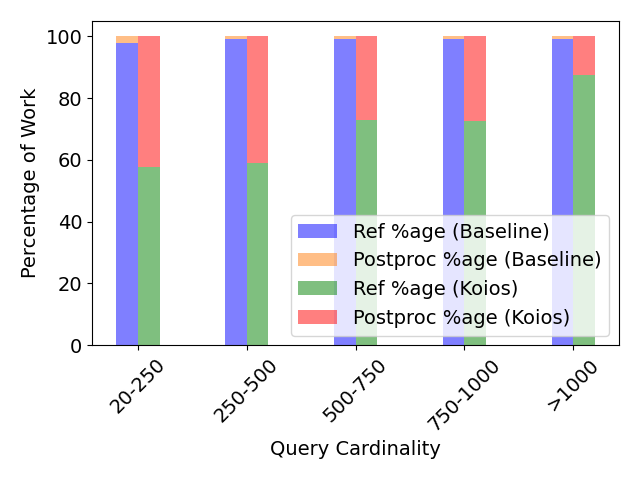} 
        	\vspace{-6mm}
        	\caption{}
            \label{fig:timeprewdc}
    \end{subfigure}
    \hfill
    \begin{subfigure}[t]{0.24\linewidth}
        	\centering
        	\includegraphics[width =\textwidth]{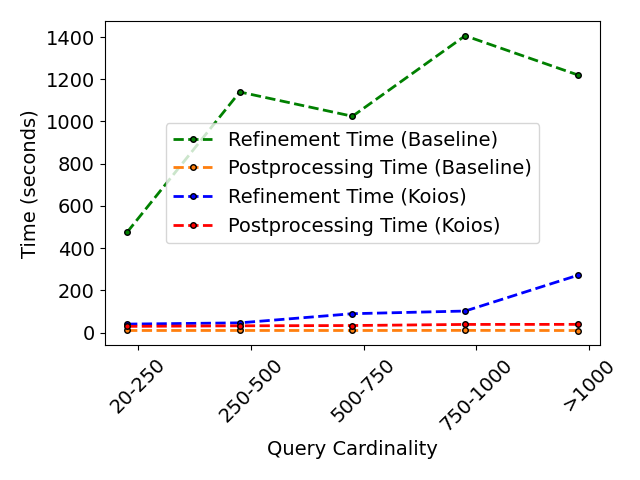} 
        	\vspace{-6mm}
        	\caption{}
            \label{fig:timebrwdc}
    \end{subfigure}
    \hfill
    \begin{subfigure}[t]{0.24\linewidth}
        	\centering
        	\includegraphics[width =\textwidth]{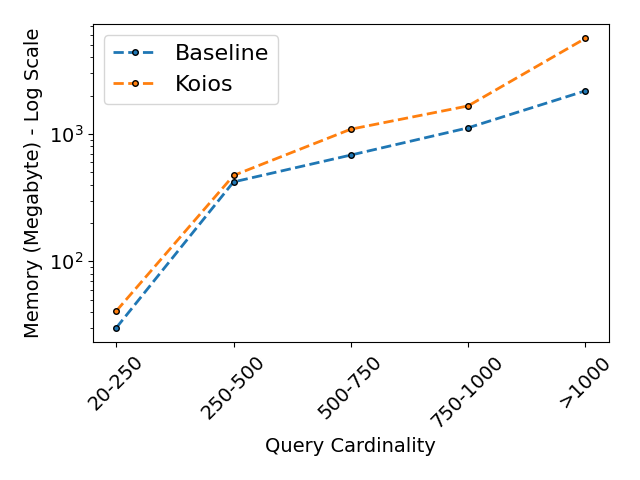} 
        	\vspace{-6mm}
        	\caption{}
            \label{fig:memwdc}
    \end{subfigure}
    \caption{\pranayrev{\camera{WDC results: (a) response time, (b), (c) phase breakdown, (d) memory footprint.}}}
\end{figure*}

\eat{\begin{figure*}[!ht]
\label{fig:odresults}
    \begin{subfigure}[t]{0.32\linewidth}
        	\centering
        	\includegraphics[width =\textwidth]{figs/new_time-tt-bs-open_data-seq.png}
        	\vspace{-6mm}
        	\caption{}
            \label{fig:timeresod}
    \end{subfigure}
    \hfill
    \begin{subfigure}[t]{0.32\linewidth}
        	\centering
        	\includegraphics[width =\textwidth]{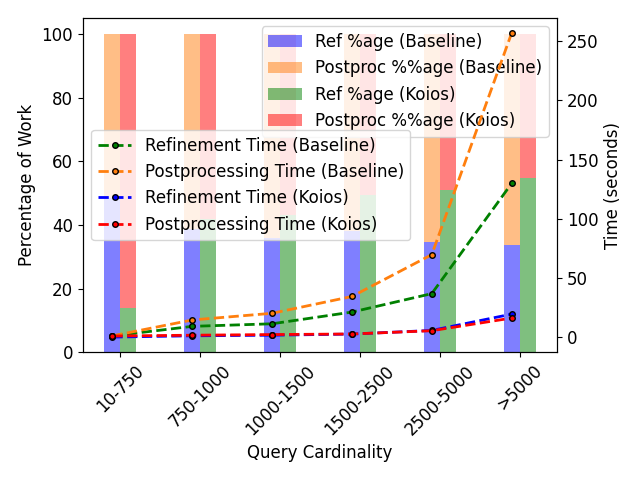} 
        	\vspace{-6mm}
        	\caption{}
            \label{fig:timepreod}
    \end{subfigure}
    \hfill
    \begin{subfigure}[t]{0.32\linewidth}
        	\centering
        	\includegraphics[width =\textwidth]{figs/new_mem-open_data.png} 
        	\vspace{-6mm}
        	\caption{Memory Footprint}
            \label{fig:memod}
    \end{subfigure}
    \caption{OpenData Results: (a) Response Time, (b) Phase Breakdown, (c) Memory Footprint}
\end{figure*}
\begin{figure*}[!ht]
\label{fig:wdresults}
\vspace{-4mm}
    \begin{subfigure}[t]{0.32\linewidth}
        	\centering
        	\includegraphics[width =\textwidth]{figs/new_time-tt-bs-wdc-seq.png} 
        	\vspace{-6mm}
        	\caption{}
            \label{fig:timereswdc}
    \end{subfigure}
    \hfill
    \begin{subfigure}[t]{0.32\linewidth}
        	\centering
        	\includegraphics[width =\textwidth]{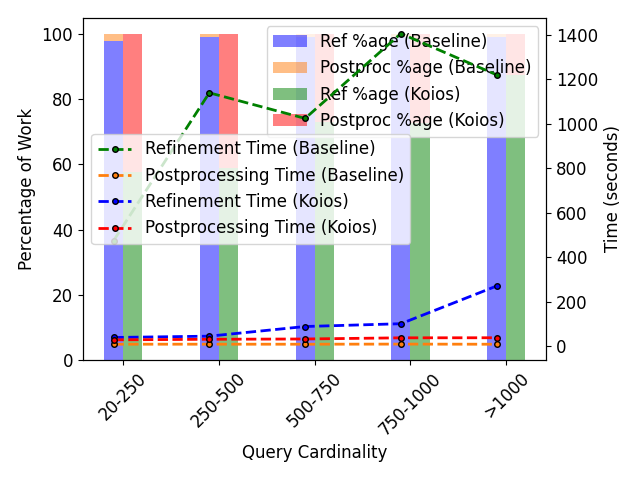} 
        	\vspace{-6mm}
        	\caption{}
            \label{fig:timeprewdc}
    \end{subfigure}
    \hfill
    \begin{subfigure}[t]{0.32\linewidth}
        	\centering
        	\includegraphics[width =\textwidth]{figs/new_mem-wdc.png} 
        	\vspace{-6mm}
        	\caption{}
            \label{fig:memwdc}
    \end{subfigure}
    \caption{WDC Results: (a) Response Time, (b) Phase Breakdown, (c) Memory Footprint}
\end{figure*}
}
\eat{\begin{figure*}[!ht]
\vspace{-4mm}
    \begin{subfigure}[t]{0.24\linewidth}
        	\centering
        	\includegraphics[width =\textwidth]{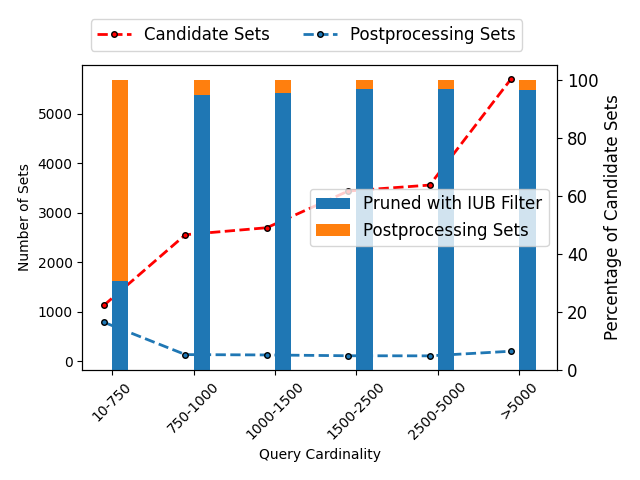} 
        	\vspace{-6mm}
        	\caption{Refinement - OpenData}
            \label{fig:rpod}
    \end{subfigure}
    \hfill
    \begin{subfigure}[t]{0.24\linewidth}
        	\centering
        	\includegraphics[width =\textwidth]{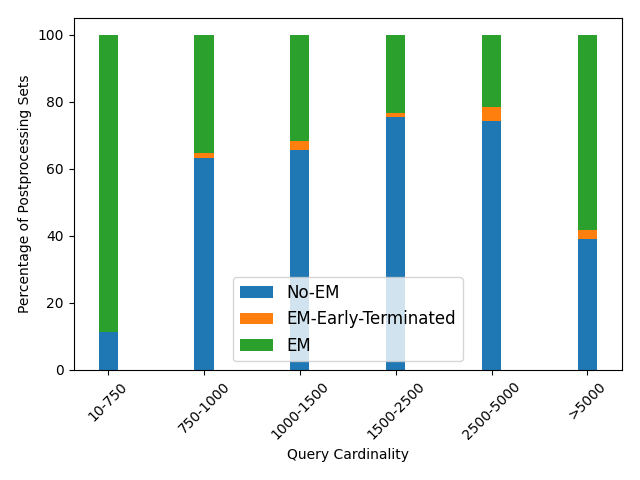} 
        	\vspace{-6mm}
        	\caption{Post-Processing -  OpenData}
            \label{fig:ppod}
    \end{subfigure}
    \hfill
    \begin{subfigure}[t]{0.24\linewidth}
        	\centering
        	\includegraphics[width =\textwidth]{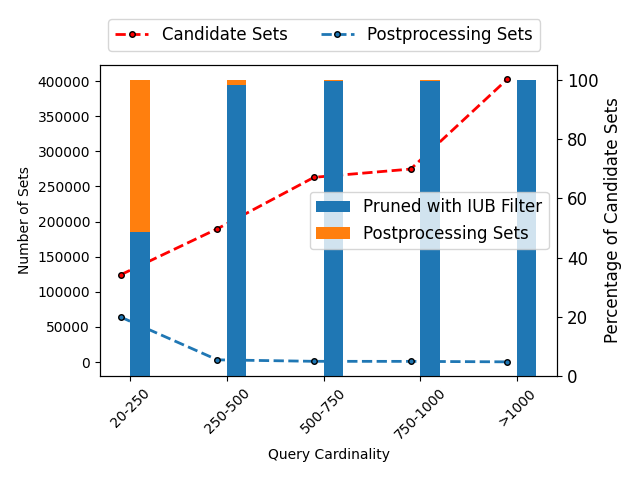} 
        	\vspace{-6mm}
        	\caption{Refinement - WDC}
            \label{fig:rpwdc}
    \end{subfigure}
    \hfill
    \begin{subfigure}[t]{0.24\linewidth}
        	\centering
        	\includegraphics[width =\textwidth]{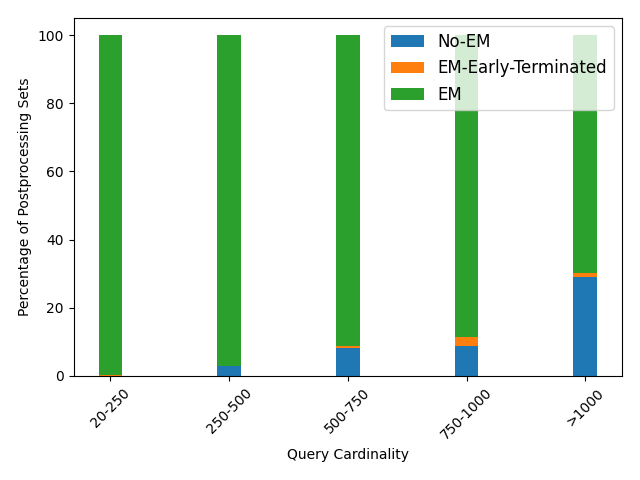} 
        	\vspace{-6mm}
        	\caption{Post-Processing -  WDC}
            \label{fig:ppwdc}
    \end{subfigure}
    \caption{Pruning Power of Filters \fnrev{changing this to a table}}
    \label{fig:5}
\end{figure*}}

\begin{table}[t]
\centering 
\caption{\camera{Average percentage of sets pruned using filters.}} \label{tbl:datasets_filters}
    \begin{tabular}{|c|c|cc|}
    \hline
     Datasets & Refinement & \multicolumn{2}{c|}{Postprocessing} \\
     & iUB-Filter & EM-Early-Terminated & No-EM \\
    \hline
    DBLP & $91\%$ & $5\%$ & $9.2\%$ \\
    OpenData & $85.5\%$ & $2.1\%$ & $54.8\%$ \\
    Twitter & $53.5\%$ & $0\%$ & $1.4\%$ \\
    WDC & $89.2\%$ & $0.9\%$ & $9.8\%$ \\  
    \hline
    \end{tabular}
\end{table}

\eat{
\begin{table*}[ht!]
\begin{minipage}{0.75\columnwidth}
\centering 
    \includegraphics[width=0.9\columnwidth]{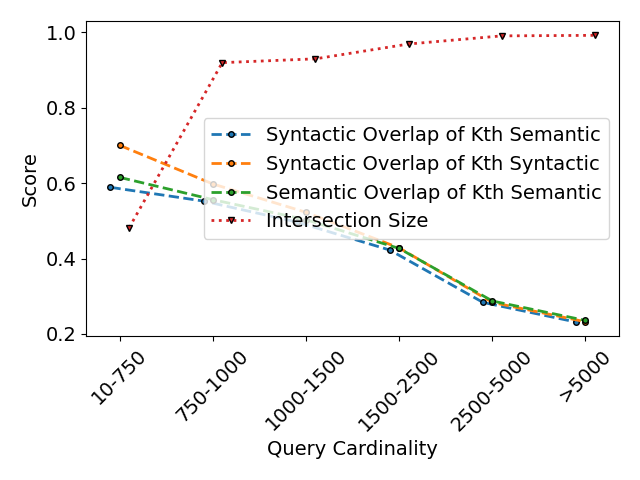}
    \captionof{figure}{\fnrev{Quality of Results}}
\label{fig:syn-sem}
\end{minipage}
\hfill
\begin{minipage}{1\columnwidth}
\centering
\captionof{table}{\fn{Average response time and memory footprint}}
\label{tbl:datasets_time}
\resizebox{\columnwidth}{!}{
    \begin{tabular}{|c|cccc|cc|}
    \hline
     & \multicolumn{4}{c|}{\algo} & \multicolumn{2}{c|}{Baseline} \\
     Datasets & Refinement & Postproc & Response & Mem & Response & Mem\\
     & (sec) & (sec) & (sec) & (MB) & (sec) & (MB)\\
    \hline
    DBLP & $0.3$ & $0.44$ & $0.83$ & $16$ & $211$  & $11$\\
    OpenData & $7.19$ & $6.9$ & $18.6$ & $69.6$ & $101$ & $102.5$\\
    Twitter & $0.2$ & $0.45$ & $0.7$ & $10$ & $518$  & $10$\\
    WDC & $109$ & $34.3$ & $147$ & $1,775$ & $1,062$ & $885$\\
    \hline
    \end{tabular}
}
\end{minipage}
\end{table*}
}

\begin{table}
\centering
\captionof{table}{\camera{Average response time and memory footprint.}}
\label{tbl:datasets_time}
\resizebox{\columnwidth}{!}{
    \begin{tabular}{|c|cccc|cc|}
    \hline
     & \multicolumn{4}{c|}{\algo} & \multicolumn{2}{c|}{Baseline} \\
     Datasets & Refinement & Postproc & Response & Mem & Response & Mem\\
     & (sec) & (sec) & (sec) & (MB) & (sec) & (MB)\\
    \hline
    DBLP & $0.3$ & $0.44$ & $0.83$ & $16$ & $211$  & $11$\\
    OpenData & $7.19$ & $6.9$ & $18.6$ & $69.6$ & $101$ & $102.5$\\
    Twitter & $0.2$ & $0.45$ & $0.7$ & $10$ & $518$  & $10$\\
    WDC & $109$ & $34.3$ & $147$ & $1,775$ & $1,062$ & $885$\\
    \hline
    \end{tabular}
}
\end{table}

\subsection{Response Time}
\label{sec:runtime}

We report the average response time, in seconds, across benchmarks for all datasets in Table~\ref{tbl:datasets_time} \pranay{(inverted index and token index construction time are excluded from \S~\ref{sec:implementation})}.\eat{Note that the response time does not take into account the time taken to construct the inverted index and \eat{Faiss}token index.}\eat{\pranay{For timed-out queries, we set the response time to the time-out limit of $2500$ seconds.}} \pranay{We do not report the time for the timed-out queries ($2500$ seconds),  therefore, we do not have enough data for some intervals of WDC and OpenData.} 
According to Table~\ref{tbl:datasets_time}, \algo achieves at least 5x speedup over the baseline across all datasets and at least 200x for DBLP and Twitter. We present additional analyses of WDC and OpenData  based on query cardinality, since these  datasets demonstrate a large set cardinality skew.

{\bf Effect of Query Cardinality: } Fig.~\ref{fig:timeresod},~\ref{fig:timepreod},~\ref{fig:timereswdc}, and~\ref{fig:timeprewdc} show the response time and the relative time spent in each phase for OpenData and WDC, respectively. The time reported is averaged over all  queries in each interval.  Because \algo processes all partitions in parallel, we report the average ratio of time spent by a partition over all queries in an interval.  We observe that the  response time increases with query cardinality, \pranay{which entails} a larger number of similar elements and candidate sets returned by $\mathcal{I}_s$. \pranay{The share  of work of WDC in the refinement  is higher than OpenData, because of its sheer number of sets and the high frequency of elements.} 

{\bf Comparing to Baseline:} Fig.~\ref{fig:timeresod} and~\ref{fig:timereswdc} show that \algo particularly outperforms the baseline for medium to large queries in OpenData and WDC. This emphasizes the pruning power of \pranay{the} filters. 
\fn{Fig.~\ref{fig:timeresod} and~\ref{fig:timereswdc} also report the number of timed-out queries which are much higher for the baseline than \pranay{for} \algo, since the baseline does not prune the large low potential sets in the refinement phase.} \pranay{\eat{However,}\algo times out for only  approximately $5\%$ of OpenData queries and approximately $20\%$ of WDC queries for small queries. This is due to the large posting lists, which result in a significant number of sets during post-processing as\eat{seen} visible \pranayrev{from Tables~\ref{tbl:opendata_pruning} and~\ref{tbl:wdc_pruning}.}\eat{in Fig.~\ref{fig:rpod} and~\ref{fig:rpwdc}.}  We also note that \algo times out for certain medium-large queries, which is due to expensive graph matching.} \pranay{In summary, the filter overhead of \algo during  refinement clearly pays off and significantly improves the runtime over the baseline. }
We remark that \algo finds $k \times \left(\text{number of partitions}\right)$ sets significantly faster than the time-out limit, whereas\eat{practically} the majority of \pranay{large} queries\eat{from the last interval} for WDC time  out for the baseline.\eat{We see that almost all of the queries from the last interval time-out for the baseline, whereas \algo finds  $k * \left(\text{number of partitions}\right)$ sets much faster than the time-out limit.}\eat{Note that for smaller queries although the response time of the baseline and \algo are almost similar for OpenData because \algo works on partitions, it finds $k * \left(\text{number of partitions}\right)$ sets as compared to the baseline that finds $k$ sets.} \pranay{Although the difference in response time of \algo and baseline is not  prominent for smaller queries of OpenData, \algo discovers $k \times \left(\text{number of partitions}\right)$  sets as opposed to the baseline.} 
\eat{Semantic overlap search by \algo has low response time across all \pranay{datasets}\eat{intervals}. The response time increases with query cardinality. This is because a large cardinality means a larger number of similar elements and a larger number of candidate sets returned by $\mathcal{I}_s$. 
A similar observation is made in Fig.~\ref{fig:timepreod} and ~\ref{fig:timeprewdc}, where \algo spends the majority of its time in the refinement phase. The time spent in post-processing becomes on par with the time spent in the refinement phase for large queries.  
Moreover, a larger query results in a more expensive graph matching.  

We see that this trend continues for WDC as well, the response time increases with query cardinality for \algo. The green bar in Fig.~\ref{fig:timereswdc} shows the percentage of queries that time out ($2,500$ seconds) for the baseline. We see that for larger queries the baseline times out for almost all queries, compared to \algo that successfully returns the top-$k$ sets for all queries.}
\eat{
For WDC, the algorithm spends more time in the refinement phase than in the post-processing phase. This is due to the sheer number of sets that the algorithm retrieves from the posting lists of WDC elements. The bounds of these sets need to be tuned in this phase. Unlike WDC, for OpenData, the algorithm spends most of its time in the post-processing and calculating exact matches. 
This is due to the overall large set cardinality of sets in OpenData. }

\fnrev{\textbf{Comparing to Fuzzy Search:} Fuzzy search techniques including Fast-Join\cite{5767865,WangLF14} and SILKMOTH~\cite{DengKMS17} cannot solve the top-$k$ overlap similarity search problem that we solve in this paper: (1) Only specific character-level similarities (Jaccard on 
\fnrev{white-space separated tokens of an element} 
and edit distance) between set \fnrev{elements} are supported; semantic \eat{token} similarity like cosine on word \eat{token} embeddings (as we use in our experiments) cannot be applied. (2) Existing fuzzy search algorithms are threshold-based: In order to retrieve the top-$k$ most similar sets to a given query, the threshold $\theta^*_k$ is required, but this threshold is not known upfront; in fact, it is one of the challenges of top-$k$ search and part of our solution to converge to this threshold quickly.} 

\fnrev{While fuzzy search algorithms do not support semantic \fnrev{element} 
similarity, our KOIOS algorithm does support all syntactic \fnrev{element} 
similarity functions. We extend the threshold-based fuzzy search to support top-$k$ search as follows: (1) \eat{Pick a random sample of $k$ sets from input collection $\mathcal{L}$ and compute the threshold $\theta$ as the minimum similarity of the query to any of the $k$ sets.}\pranayrev{Pick the threshold $\theta$ as the minimum $\theta^*_k$ amongst all the queries. Note this gives SILKMOTH an advantage as we pass the true value of $\theta^*_k$} (2) Compute the fuzzy search with threshold $\theta$ and select the top-$k$ most similar sets from the query result (by maintaining a top-$k$ priority queue).}

\fnrev{We focus our study on SILKMOTH, which was shown to widely outperform Fast-Join~\cite{DengKMS17} \fnrev{and consider OpenData and WDC datasets. Since sets in these datasets are extracted from tables, the majority of elements consist of very few words,  which results in a zero edge weight for non-identical elements in SILKMOTH. Therefore, in our experiments, we consider Jaccard on $3$-grams representation of each element as an element similarity for both \algo and SILKMOTH.} 
To generate the token stream, we  precompute  elements in the vocabulary that are similar to each query element with the Jaccard similarity threshold of $\alpha = 0.8$, using the set similarity join techniques~\cite{MannAB16}. It takes $8$ seconds to compute the token stream for the benchmark.}

\fnrev{We compare two versions of SILKMOTH: The first version, SILKMOTH-semantic, adapts the SILKMOTH algorithm to cover most of the functionality of \algo\footnote{\fnrev{SILKMOTH-semantic requires a \emph{metric} token similarity function since the triangle inequality is leveraged. \algo requires only symmetry and can also deal with cosine similarity on token embeddings, which is not a metric.}}; this adaption was suggested by the original authors as a generic search framework and excludes \eat{$q$-gram}\pranayrev{similarity function} specific filters~\cite{DengKMS17}. The second version, SILKMOTH-syntactic, uses all indexes and filters, including those that are applicable only to \pranayrev{particular similarity functions}. We run the algorithms on \pranayrev{54 queries randomly sampled from the benchmark queries of OpenData to make sure we evaluate on small, medium, and large sets} and measure response time. The average response time of \algo, SILKMOTH-syntactic, and  SILKMOTH-semantic are $72$, $141$, and $400$ seconds, respectively.   \algo outperforms \pranayrev{both SILKMOTH-semantic, and SILKMOTH-syntactic by an order of magnitude $6$x(timed-out), and $2$x respectively on all query size ranges. SILKMOTH produces signatures from the set elements and utilizes them to help decrease the candidate space. The number of signatures increases when the number of set elements increases (here by  splitting into $q$-grams) 
and SILKMOTH-syntactic times out due to the sheer number of viable candidates. 
\algo outperforms SILKMOTH-syntactic because it operates on an ordered stream of set element pairings based on  similarity  and is thus unaffected by the number of elements. \algo outperforms SILKMOTH-semantic because the pruning power of SILKMOTH is highly dependent on filters that are specialized for certain similarity functions.}} 

\subsection{\pranayrev{Pruning Power of Filters}}
\label{sec:prunpower}

\pranay{ Table~\ref{tbl:datasets_filters} reports the mean pruning power of different filters used in both phases for all datasets and \pranayrev{Tables~\ref{tbl:opendata_pruning} and~\ref{tbl:wdc_pruning}}\eat{Figure~\ref{fig:5}} zooms into the pruning power for OpenData and WDC across query cardinality intervals.} 

{\bf Refinement Phase: } \eat{The majority of sets are pruned by the iUB filter during the refinement phase (Table~\ref{tbl:datasets_filters}). 
As shown in Fig.~\ref{fig:rpod} and~\ref{fig:rpwdc}, the total number of candidate sets and the blue lines show the number of sets that require  post-processing.} In OpenData and WDC, the number of candidate sets increases with query cardinality because  more posting lists from $\mathcal{I}_s$ \pranay{must be read}. The pruning power of the iUB-filter increases with the increase of query cardinality. This is because of two reasons, \fn{first, there are a significant amount of small cardinality sets in OpenData and WDC (Table~\ref{tbl:datasets}),}\eat{first from \S~\ref{sec:datasets} we know that for OpenData and WDC\eat{can see from Fig.~\ref{fig:setdistod}} that there are a significant amount of small cardinality sets, indicating that we end up with a lot of small cardinality candidate sets from} \pranay{that are returned by} the posting lists. Second, the iUB of each set is calculated based on the number of remaining elements, thus candidate sets with smaller cardinality relative to the query set have a much smaller iUB as compared to the $\theta_{lb}$ and are hence pruned. 
As shown in the \pranayrev{Tables~\ref{tbl:opendata_pruning} and~\ref{tbl:wdc_pruning}}\eat{bar plots of~\ref{fig:rpod} and~\ref{fig:rpwdc}}\eat{and~\ref{fig:refinedblptwitter}}, although the number of candidate sets increases with query cardinality, the fraction that requires post-processing by \algo decreases with query cardinality. For example, \algo requires the post-processing of \eat{fewer}\pranay{less} than 5\% of candidate sets for large queries for WDC.

\begin{figure*}[!ht]
\begin{subfigure}[t]{0.24\linewidth}
        \centering
        	\includegraphics[width =\textwidth]{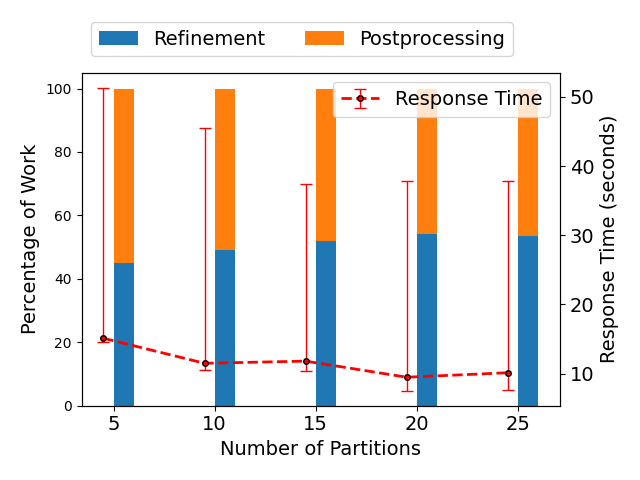} 
        	\vspace{-6mm}
        	\caption{}
            \label{fig:parts}
    \end{subfigure}
    \hfill
    \begin{subfigure}[t]{0.24\linewidth}
        	\centering
        	\includegraphics[width =\textwidth]{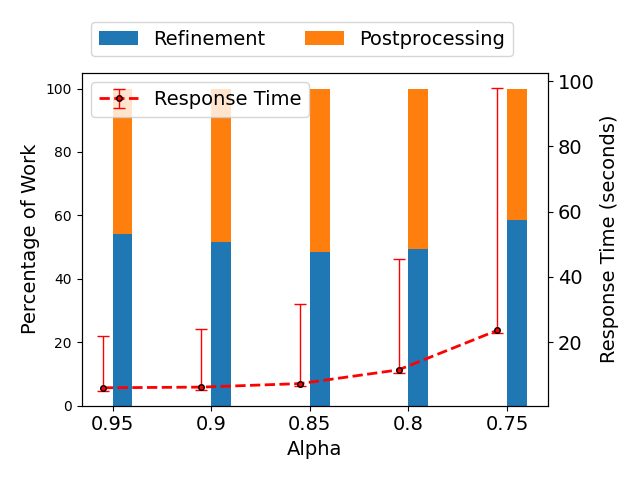} 
        	\vspace{-6mm}
        	\caption{}
            \label{fig:alpha}
    \end{subfigure}
    \hfill
    \begin{subfigure}[t]{0.24\linewidth}
        	\centering
        	\includegraphics[width =\textwidth]{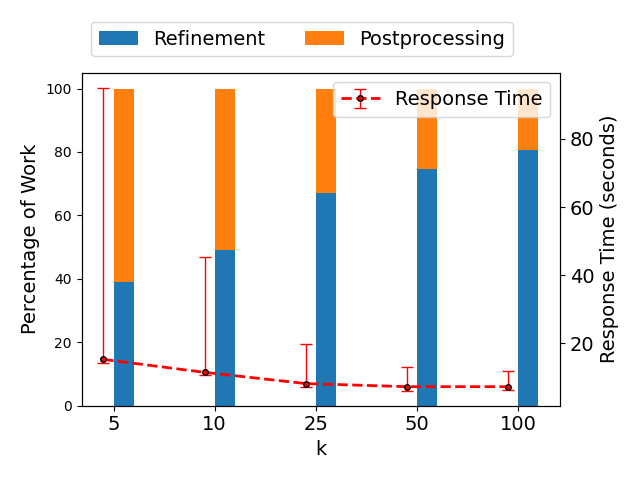} 
        	\vspace{-6mm}
        	\caption{}
            \label{fig:k}
    \end{subfigure}
    \hfill
    \begin{subfigure}[t]{0.24\linewidth}
        	\centering
        	\includegraphics[width =\textwidth]{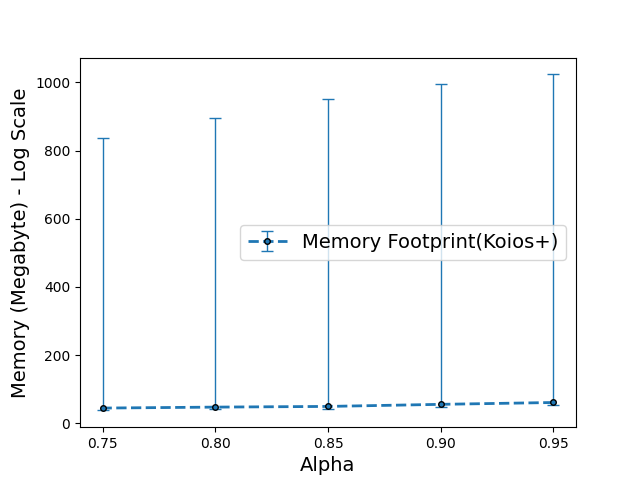}
        	\vspace{-6mm}
        	\caption{}
            \label{fig:memodalpha}
    \end{subfigure}
    \caption{\camera{Parameter analysis of \algo on OpenData: Time vs. (a) number of partitions, (b) element similarity ($\alpha$), and (c) result size, (d) memory footprint vs. alpha.}}
    \label{fig:params}
\end{figure*}

{\bf Post-processing Phase:} \fn{Table~\ref{tbl:datasets_filters} shows that the No-EM filter demonstrates a higher pruning power than EM-Early-Terminated, e.g.,  pruning more than half of sets for OpenData.  Note that the reported percentages refer to the sets that are not filtered in the refinement phase.} From \eat{Figures~\ref{fig:ppod} and~\ref{fig:ppwdc}}\pranayrev{Tables~\ref{tbl:opendata_pruning} and~\ref{tbl:wdc_pruning}} we observe that the combination of No-EM and EM-Early-Terminated have the highest pruning power for large queries. For OpenData, the combination of these two filters prunes more sets than in WDC. This is because an exact calculation of semantic overlap during the post-processing increases $\theta_k$ and results in pruning many sets.

\begin{table}[h]
\centering
\captionof{table}{\fnrev{\camera{OpenData: \#sets pruned by filters.}}}
\label{tbl:opendata_pruning}
\resizebox{\columnwidth}{!}{
    \begin{tabular}{|c|cc|ccc|}
    \hline
     & \multicolumn{2}{c|}{Refinement Phase} & \multicolumn{3}{c|}{Postprocessing Phase} \\
     Query Card. & Candidate & iUB-Filtered & No-EM & EM-Early & EM\\
      & Sets& & & Terminated & \\
     \hline
     $10-750$ & $1132$ & $345$ & $88$ & $0$ & $699$\\
     $750-1000$ & $2557$ & $2422$ & $85$ & $2$ & $48$ \\
     $1000-1500$ & $2699$ & $2571$ & $83$ & $4$ & $41$ \\
     $1500-2500$ & $3440$ & $3328$ & $84$ & $2$ & $26$ \\
     $2500-5000$ & $3560$ & $3451$ & $82$ & $4$ & $23$ \\
     $>5000$ & $5706$ & $5502$ & $79$ & $5$ & $120$ \\
    \hline
    \end{tabular}
}
\end{table}
\begin{table}[h]
\captionof{table}{\fnrev{\camera{WDC: \#sets pruned by filters.}}}
\label{tbl:wdc_pruning}
\centering
\resizebox{\columnwidth}{!}{
    \begin{tabular}{|c|cc|ccc|}
    \hline
     & \multicolumn{2}{c|}{Refinement Phase} & \multicolumn{3}{c|}{Postprocessing Phase} \\
     Query Card. & Candidate  & iUB-Filtered & No-EM & EM-Early & EM\\
     & Sets & & & Terminated & \\
     \hline
     $20-250$ & $124,217$ & $60,196$ & $74$ & $80$ & $63,867$ \\
     $250-500$ & $189,665$ & $186,512$ & $90$ & $3$ & $3,060$ \\
     $500-750$ & $262,947$ & $261,901$ & $85$ & $6$ & $953$ \\
     $750-1000$ & $274,695$ & $273,743$ & $83$ & $26$ & $843$ \\
     $>1000$ & $402,622$ & $402,332$ & $84$ & $3$ & $203$ \\
    \hline
    \end{tabular}
}
\end{table}

\subsection{Memory Footprint}
\label{sec:memory}

\pranay{\algo is an in-memory algorithm.} We report the average memory across benchmarks for all datasets in Table~\ref{tbl:datasets_time}. The extended version of the paper contains an in-depth analysis of \algo memory footprint. Fig.~\ref{fig:memod} and~\ref{fig:memwdc} show the memory footprint of  
\algo \pranay{and the baseline} for OpenData and WDC. The reported values are the average memory footprint of data structures over successful queries in each interval. Note that some data structures such as inverted index, token index, and top-$k$ lists have fixed sizes for all query intervals. To have the whole picture of the memory footprint, the numbers reported in these plots and Table~\ref{tbl:datasets_time} are the sum of the footprint of data structures  used in the refinement phase and post-processing phase, although the data structures used in the refinement phase, except the top-$k$ LB list ($\mathcal{L}_{lb}$), are freed up at the end of the phase. Table~\ref{tbl:datasets_time} shows that \algo' memory utilization is comparable to the baseline.

{\bf Effect of Query Cardinality:} In Fig.~\ref{fig:memod} and~\ref{fig:memwdc}, we observe that the memory footprint increases linearly with the query cardinality. This can be explained by\eat{: 1)} the linear increase in the number of candidate sets (\eat{Figures~\ref{fig:rpod} and~\ref{fig:rpwdc}}\pranayrev{Tables~\ref{tbl:opendata_pruning} and~\ref{tbl:wdc_pruning}}), thus, the size of the query-dependent data structures: token stream, upper-bound buckets, lower-bound data structure, and priority queues increases as the query cardinality increases.\eat{, and 2) the quadratic increase in the size of the similarity matrix.} 
\eat{\fn{quadratic. increase in matrix size -> linear increase in mem footprint?}}

{\bf Comparing to Baseline: }Fig.~\ref{fig:memod} and~\ref{fig:memwdc} show the average memory footprint of the baselines and \algo.\eat{We see that for} \pranay{For} small and medium queries for both OpenData and WDC, the memory footprint is similar.
\eat{For large queries for OpenData \algo takes up less memory as compared to \algo.}\fn{For large queries, \algo takes up less memory for OpenData compared to WDC.} This is because, the iUB-Filter prunes most of the candidate sets, hence reducing the size of the post-processing data structures. Note that for\eat{WDC} \pranay{the baseline} we do not have enough data for\eat{the baselines} \pranay{WDC} on large queries as almost all queries time  out.  
\eat{We also see from Fig.~\ref{fig:memodalpha} that increasing $\alpha$ results in a slight increase in the memory footprint for \algo. This is because, by increasing alpha, we get a smaller token stream and a decrease in the number of candidate sets. With fewer sets being processed we converge to a much smaller value of $\theta_{lb}$ which results in more sets reaching post-processing and this results in larger post-processing data structures, hence, the increase in the memory footprint.}

\eat{
\pranayrev{\subsection{\todo{TODO:} Comparing to SILKMOTH}
We compare our framework to SILKMOTH, a specialized system capable of detecting related set pairs in a collection of sets~\cite{DengKMS17}. Although SILKMOTH supports a variety of similarity measures, their novel filtering methods that reduce the number of candidate sets rely on the requirement that the similarity function satisfies the triangle inequality, which limits the measures for string matching to Jaccard and Edit similarity. The filters utilized by \algo, on the other hand, are independent of the properties of the similarity function, therefore our performance is not limited by the similarity measure. For example, unlike \algo, SILKMOTH does not enable cosine similarity between word embeddings to find related sets. Furthermore, \algo is not restricted to string matches; rather, it can aid in the discovery of related sets if there exists a similarity function between set elements as defined in Definition~\ref{def:semoverlap}.} 

\pranayrev{To demonstrate Koios' efficiency and generalizability, we compare it to SILKMOTH with the similarity specific filters turned off on 54 randomly picked benchmark queries from OpenData's queries over all intervals to have a fair representation. We use the same similarity function as SILKMOTH : Jaccard Similarity on q-grams of elements, where $q = 3$, and elements are the set elements. As per our knowledge, there doesn't exist an efficient index, that could help generate the token stream on-the-fly for the above similarity metric, so we precompute the token streams for each query with the similarity threshold, $\alpha = 0.8$, using~\cite{MannAB16}. Since SILKMOTH performs a threshold based search and \algo does a top-$k$ discovery, to have a fair comparision we run SILKMOTH under the problem setting of RELATED SET DISCOVERY with $\alpha = 0.8$ and the relatedness threshold $\delta = 0.02$ which is the minimum $\theta^*_k$ among the $54$ benchmark queries.} \todo{Add discussion based on results.}}

\subsection{Quality of Results}
\label{sec:syntactic}
\pranay{To evaluate the quality benefit of semantic overlap, we compare the top-$k$ semantic search results with the top-$k$ syntactic search results on OpenData, using the vanilla set overlap.  \fn{Fig.~\ref{fig:syn-sem} reports the scores for the $k$-th set in the  top-$k$ lists. Comparing the syntactic overlap of the $k$-th set in the  top-$k$ syntactic list with that of the   top-$k$ semantic list,  we observe that semantic overlap finds sets with lower syntactic overlap (fewer exact matching elements) but higher semantic overlap (more elements with semantic similarity). These are sets that often  contain syntactically dirty elements, for example ({\tt squirrel}, {\tt squirrell}) and ({\tt konstantine}, {\tt konstantin}), or syntactically mismatching elements yet  semantically similar elements, for example ({\tt Leeds}, {\tt Sheffield}). } 
\fn{We also investigated the sets returned by the two searches and report the size of the intersection of results. Fig.~\ref{fig:syn-sem} shows that semantic overlap finds sets that cannot be otherwise discovered by the vanilla overlap.\eat{This is more obvious for small and medium size sets, in particular, for the smallest interval almost 50\% of results are nowhere to  be found by the vanilla overlap.}} \pranay{In particular, for the smallest interval vanilla overlap misses $50\%$ of the results.} \pranayrev{This shows how semantic overlap can help find sets that would not be part of the top-$k$ result if only syntactic overlap was considered. Note that \algo returns an exact solution as long as the index returns exact results.}  
\eat{The intersection size of the top-k sets for both searches for OpenData over all intervals is bigger than 0.5, and the semantic score for the sets is higher than the syntactic results. This demonstrates the validity of our algorithm's results as well as the fact that we find more matches than the syntactic overlap, resulting in a higher semantic score.}}

\subsection{Analysis of Parameters}
\label{sec:params}

Our algorithm uses three parameters for search: 1) an element similarity threshold ($\alpha$), 2) the number of data partitions, and 3) a user-provided number of result sets ($k$).  
We performed an empirical study of the impact of these  parameters on the response time of \algo. For this set of experiments, we choose 100  queries from the benchmark of OpenData at random. Since the set cardinality in data repositories follows a Zipfian distribution~\cite{ZhuNPM16,ZhuDNM19}, random sampling from  benchmark intervals prevents from having a benchmark that is heavily biased to smaller query sets.\eat{Plots of} Fig.~\ref{fig:params} shows the average response time and the breakdown of the ratio of refinement and post-processing time of queries for various parameter values. 

\begin{figure}[tbp]
\centering 
    \includegraphics[width=0.7\linewidth]{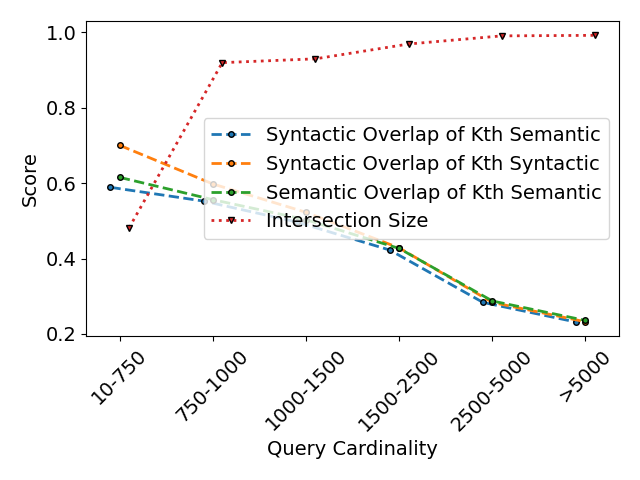}
    \caption{\fnrev{\camera{Comparison of vanilla and semantic overlap.}}}
\label{fig:syn-sem}
\end{figure}

{\bf Number of Partitions: } In Fig.~\ref{fig:parts}, we fix $k=10$ and $\alpha=0.8$ and vary the number of partitions. The response time decreases\eat{with the increase in} \pranay{as} the number of partitions \pranay{increases}. This is because a larger number of partitions\eat{means} \pranay{results in} a smaller number of sets to process\eat{in a} \pranay{per} partition. Since partitions\eat{run the search algorithm} \pranay{are processed} in parallel, the response time\eat{is smaller when we have more} \pranay{decreases with more} partitions.\eat{Note that since} \pranay{Since} sets are randomly assigned to partitions, partitions have the same expected number of sets. The top-$k$ result of all partitions\eat{needs to} \pranay{are} merge-sorted after all partitions finish.\eat{With a large number of partitions, the algorithm need to merge more top-$k$ results.} However, the merging cost is negligible \pranay{compared to the overall runtime.}\eat{at least for $k=10$.} 

In addition to the average response time, Fig.~\ref{fig:parts} reports the average percentage of work in the refinement  and post-processing of a partition. 
One interesting observation is that as the number of partitions increases, the percentage of post-processing \pranay{sets} becomes smaller. This is because, with a large number of partitions, more sets are considered per time unit across all partitions and $\theta_{lb}$ grows quickly. \pranay{As a result, iUB-Filter has higher pruning power.}\eat{for a larger  number of partitions.}

{\bf Element Similarity Threshold: }In Fig.~\ref{fig:alpha}, we fix $k=10$, the number of partitions to 10, and vary the value of $\alpha$. We observe that the higher $\alpha$, the smaller the average response time for a query. A smaller $\alpha$ means more elements in the vocabulary are potentially considered for matching with element queries, and as a result, more candidate sets are considered by the algorithm.\eat{The pruning power of refinement filters is independent of the value of $\alpha$. However, because the cost of graph matching grows with the number of edges in a bipartite graph, larger $\alpha$ means on average more edges in the graph and on average higher matching time.} \pranay{While the pruning power of refinement filters is independent of the value of $\alpha$, the cost of graph matching grows with the number of edges in a bipartite graph: smaller $\alpha$  values result in more edges in the graph and therefore higher matching time. Fig.~\ref{fig:memodalpha} also shows that increasing $\alpha$ results in a slight increase in the memory footprint for \algo. This is because by increasing $\alpha$, we get a smaller token stream and a decrease in the number of candidate sets.\eat{With fewer sets being processed we} We converge to a\eat{much} smaller value of $\theta_{lb}$ which results in more sets reaching post-processing; this results in larger post-processing data structures and hence the increase in the memory footprint.}

{\bf Result Set Cardinality: }In Fig.~\ref{fig:k}, we fix the number of partitions to 10 and $\alpha=0.8$ and vary the value of $k$. 
For a \pranay{given} value of $k$ in this plot, each partition runs a top-$k$ search and the final results from partitions are merged, i.e. for $k=50$, we merge ten  top-$50$ lists returned from  ten partitions. 
The observation of a decrease in average response time with the increase of $k$ is counter-intuitive since higher $k$ means lower $\theta_{lb}$ and lower pruning power of the iUB-Filter. However, the  response time decreases because the average post-processing work decreases with the increase of $k$. 

\eat{The  behavior of the response time and the  post-processing work can be explained by considering  the number of sets being considered and those that reach post-processing.} \pranay{The response time and the post-processing work behavior can be explained by looking at the number of sets considered and those that reach post-processing.} With a large $k$, since the iUB-filter prunes a lot of sets, those that reach the post-processing phase end up in the top-$k$ result, and many of them are filtered using the EM-Early-Terminated-Filter and are quickly added to the top-$k$ result. 

\section{Related Work}
\label{sec:relatedwork}

\fnrev{{\bf{\em (Fuzzy) Set Similarity Search}} The majority of \camera{works}  in set similarity search take into account syntactic measures like containment and Jaccard on sets of tokens  
and use a threshold search~\cite{jia2018survey,MannAB16}. 
To avoid computing pairwise \eat{similarity computation}\camera{similarities} of set, the common technique is to apply  a filter-verification paradigm with a core step being  prefix-filtering 
~\cite{DBLP:conf/www/BayardoMS07,DBLP:conf/icde/ChaudhuriGK06,DBLP:conf/www/BayardoMS07}. 
PPJoin extends prefix filtering by incorporating a positional filtering technique that uses token ordering information to reduce candidate sizes even further~\cite{DBLP:journals/tods/XiaoWLYW11}.  There has also been work on partition-based search, where the goal is to partition the data to allow for faster search, for example, SSJoin and GreedyPlus~\cite{arasu2006efficient,deng2015efficient}.}

\eat{
{\bf{\em Fuzzy Set Similarity Search}} To deal with misspellings and dirty sets, given a query set, fuzzy set similarity compares candidate sets based on the maximum bipartite graph matching. Yu et al. describe numerous fuzzy string search algorithms, and we see that the majority of the techniques rely on some prefix filtering method and so are dependent on the similarity function~\cite{Yu:2016}. There are two types of work in this area: algorithms that estimate the matching score and algorithms that strive to reduce the number of candidate sets that require matching computation~\cite{DengLFL13,WangLDZF15}. Agarwal et al.~\cite{AgrawalAK10}, for example, propose a token transformation-based containment metric that approximates the matching score. Algorithms that prune candidate sets, for example, Fast-join~\cite{5767865}, SILKMOTH~\cite{DengKMS17}, employ filters based on character-based similarity functions. Fast-join uses a token-based signature and edit distance based weights to assist with graph matching. 
\algo differs from the preceding algorithms in that it works for arbitrary (semantic) similarity functions rather than character-based  similarity functions. }  

{\bf{\em Table Join Search}} \eat{The majority of}\camera{Most existing} join search techniques consider equi-join and use (normalized) cardinality of the overlap of  sets of  attribute values as joinability measure~\cite{ZhuDNM19,FernandezMNM19,DengYS0L019,ZhuNPM16,BogatuFP020,CasteloRSBCF21}. JOSIE combines  filtering techniques from the set similarity search literature to solve the join search problem using vanilla  overlap~\cite{ZhuDNM19}. A common way to obtain  approximate join search results is to construct an LSH index on set signatures  generated using hash functions such as MinHash~\cite{Broder97,ZhuNPM16,FernandezMNM19}.

\fnrev{{\bf{\em Semantic Techniques}} PEXESO is a threshold-based set similarity search technique based on an extension of the vanilla overlap~\cite{DongT0O21}. 
In PEXESO, two elements are considered as ``matching'' if they have a  similarity  greater than a user-specified threshold. The set similarity is then defined  as the number of elements in a query that has at least one matching element in a candidate set, normalized by the query cardinality. This implies that elements  can participate in many-to-many matchings\eat{, making} \camera{such that} the vanilla overlap \eat{not}\camera{cannot be expressed as} a special case of the proposed measure. 
SEMA-JOIN takes two sets of values from join columns as input and produces a predicted join relationship (many-to-one join on cell values)~\cite{HeGC15}. 
To do so, SEMA-JOIN finds the join relationship that maximizes \camera{the} aggregate pairwise semantic  correlation. Relying on a big table corpus (100M tables from the web), the statistical co-occurrence is used as a proxy for semantic correlation. The intuition is that two values  can be joined semantically  (e.g., {\tt GE} and {\tt General Electric})\eat{,} 
if there exists significant statistical co-occurrence of these values 
in the same row in the corpus (row co-occurrence score). Moreover, \eat{since} \eat{Similarly} pairs of pairs are required to be semantically compatible, i.e., they  \camera{should} co-occur in the same columns in the corpus (column-level co-occurrence). Finally,  to join two columns\camera{,} SEMA-JOIN aims at  maximizing the aggregate pairwise column-level and row-level correlation. Unlike  \algo\camera{,} which solves the search problem of finding sets that can be joined with a query set, SEMA-JOIN  finds the best way of joining column elements after discovery.}

\eat{
\noindent{\bf{\em Overlap Search}} SILKMOTH also studies threshold-based join search. SILKMOTH  measures the relatedness of two sets with the maximum graph matching in the bipartite graph of sets~\cite{DengKMS17}. To determine similarity of two values (an edge weight), SILKMOTH considers token-based Jaccard similarity of values and character-based edit similarity. Unlike SILKMOTH, our definition of semantic overlap works with any similarity function. 
To perform the join search, SILKMOTH creates a signature for each set. The signature of a set $C$ is the smallest subset of tokens from $C$ such that if another set $Q$ does not share any tokens with $C$’s signature, $C$ and $Q$ cannot be joined.  
To generate signatures, SILKMOTH leverages the threshold of joinability.  The threshold on set relatedness gives a lower bound on the maximum matching score.
Using the signature, SILKMOTH selects potentially related candidate sets from the repository. 
Once the candidates are found because SILKMOTH uses Jaccard and edit similarity functions, it can take advantage of the triangle inequality in the bipartite maximum matching to speed up the processing of candidate sets for eliminating false positives. 
}

\section{Conclusion}

\eat{To enable semantic\eat{table join discovery,} similarity set search, we defined the semantic overlap measure which is the maximum matching score of the bipartite graph built on the sets of values of two columns, using arbitrary similarity functions that assign weights to value pairs. We solve the top-$k$ semantic overlap search problem using a two-phase algorithm augmented with filters. Our experiments show that the proposed algorithm has an overall low response time and low memory footprint on four data repositories with different characteristics. In particular, for medium to large sets less than 5\% of candidate sets need post-processing and more than half of those sets are further pruned \pranay{without requiring the expensive} graph matching. 
In our future work, we plan to \eat{incorporate more intelligent partitioning algorithms based on estimations of work required for processing sets.}\pranay{\eat{We also intend to} expand the semantic overlap to instances\eat{where there exists} with many-to-1 mappings, to cover noise or spelling variations within the query, for example, {\tt United States of America} and {\tt United States} can both be mapped to {\tt USA} with equal similarity.} Furthermore, we plan to evaluate our algorithm with various similarity functions.} 

\pranayrev{We defined the semantic overlap measure.  
To solve the top-$k$ semantic overlap search problem, we introduced \algo, an exact and efficient filter-verification framework with powerful and cheap-to-update filters that decrease the graph-matching computation to less than $5\%$ of the candidate sets. We demonstrated that \algo has a low response time and memory footprint \eat{across}\camera{in experiments on} four different datasets.  
In our future work, we plan to expand the semantic overlap to instances\eat{where there exists} with many-to-1 mappings\eat{,} to cover noise or spelling variations within the query, for example, {\tt United States of America} and {\tt United States} can both be mapped to {\tt USA} with equal similarity.}

\eat{\textbf{Acknowledgements.} This research was partially funded by the Austrian Science Fund (FWF) P\,34962. For the purpose of open access, the authors have applied a CC BY public copyright licence to any Author Accepted Manuscript version arising from this submission.}
\section*{Acknowledgements} 
This research was partially funded by the Austrian Science Fund (FWF) P\,34962. For the purpose of open access, the authors have applied a CC BY public copyright licence to any Author Accepted Manuscript version arising from this submission. It was also supported in part by the National Science Foundation under grant 2107050. 

\clearpage
\newpage
\bibliographystyle{IEEEtran}
\bibliography{IEEEabrv,main}
\end{document}